\documentclass[11pt]{article}
\usepackage[T1]{fontenc}
\usepackage[utf8]{inputenc}
\usepackage[british]{babel}
\usepackage{lmodern,amsmath,amssymb,amsthm,mathtools,microtype}
\usepackage[a4paper,margin=27mm]{geometry}
\usepackage{booktabs,enumitem,graphicx,tikz,pgfplots}
\usepackage{placeins}
\usetikzlibrary{arrows.meta,positioning}
\pgfplotsset{compat=1.18}
\definecolor{ink}{HTML}{203447}
\definecolor{plotblue}{HTML}{176B91}
\definecolor{plotorange}{HTML}{BD6423}
\usepackage[colorlinks=true,linkcolor=ink,citecolor=ink,urlcolor=plotblue]{hyperref}
\setlist{nosep}
\newtheorem{theorem}{Theorem}[section]
\newtheorem{lemma}[theorem]{Lemma}
\newtheorem{proposition}[theorem]{Proposition}

\theoremstyle{remark}

\DeclareMathOperator{\TV}{TV}
\DeclareMathOperator{\Var}{Var}

\newcommand{\SR}{\mathrm{SR}}
\newcommand{\EA}{\mathrm{EA}}
\newcommand{\NS}{\mathrm{NS}}
\newcommand{\E}{\mathbb E}

\newcommand{\id}{\mathrm{id}}

\hypersetup{pdftitle={Channel dichotomies},pdfauthor={Aadil Oufkir and Mario Berta},
pdfsubject={Second-order channel interconversion under non-signalling superchannels}}
\title{Channel dichotomies}
\author{Aadil Oufkir\textsuperscript{1}\qquad Mario Berta\textsuperscript{2}\\[0.6em]
\small\textsuperscript{1}College of Mathematical Sciences \& Intelligent Systems, Mohammed VI Polytechnic University,\\
\small Rocade Rabat-Sal\'e, Technopolis, Morocco\\[0.4em]
\small\textsuperscript{2}Institute for Quantum Information, RWTH Aachen University,\\
\small Aachen, Germany}
\date{}

\begin{document}
\maketitle
\begin{abstract}
We study channel interconversion under non-signalling superchannels, where the asymptotic conversion rate is given by the ratio of the channel capacities.
Our main result determines the second-order expansion of the optimal
target length. Namely, when each channel has a single capacity-achieving
dispersion, the coefficient is given by the inverse sesquinormal
distribution, with parameters determined by the capacities and dispersions.
The result extends state-dichotomy theory to channels: one superchannel
approximates the target uniformly over inputs and maps every reference
channel whose output is independent of its input to another such channel.
As in state dichotomies, there are finite-error resonances where the
second-order term vanishes, leaving an $O(\log n)$ remainder.
Crucially, our new characterization cannot be obtained by concatenating channel
coding and simulation protocols [Cao {\it et al.}~IEEE
Trans.~Inf.~Theory (2024)]. Instead, the proof combines comparisons of
binary hypothesis testing experiments with uniform Gaussian estimates
and a construction that uses symmetry within type classes to satisfy
the non-signalling constraints. Shared randomness assistance already suffices at
first order, but we show that its second-order behaviour can differ
strictly from the non-signalling case.
\end{abstract}

\section{Introduction}
\label{sec:introduction}
\noindent\textbf{AI disclaimer.} {\it Based on previous tight second-order results for the interconversion of specific channels, the general characterization given here was developed with ChatGPT (5.6 Sol and 6 Astra). The manuscript proposes a solution to an open problem in quantum information theory that we believe is of interest to the community. We also see an urgent need for a community discussion of appropriate formats for the communication of AI-generated scientific content.}

\medskip
Channel coding \cite{shannon1948communication} uses a noisy channel
to transmit messages reliably. Channel simulation
\cite{bennett2002reverse,cao2024channel} uses noiseless communication
to reproduce a noisy channel. Channel interconversion combines these
two tasks \cite{cubitt2011zero,oufkir2026interconversion}.
Given a resource channel $T$ and a target channel $W$, we ask how many
uses of $W$ can be simulated using $n$ uses of $T$, at a prescribed error.

With unlimited shared randomness, a channel code for $T$ followed by
a simulation protocol for $W$ achieves the first-order rate $C_T/C_W$
\cite{shannon1948communication,bennett2002reverse,cao2024channel}.
The reverse conversion has the reciprocal rate. Thus capacity
determines reversibility at first order. At finite blocklength, however,
optimizing coding and simulation separately need not give the best
interconversion protocol. This approach first transmits information
reliably and then adds noise. A direct protocol can instead use the
noise already present in the resource. For example, a channel
simulates itself exactly by direct use. The capacity ratio alone
therefore does not determine the second-order correction.

\subsection{Overview of main results}
We characterize this correction under \emph{non-signalling superchannels}.
The assistance is a non-signalling channel $N(x,z\mid s,y)$, also
called a non-signalling box. The sender supplies $s$ and receives $x$.
The receiver supplies $y$ and receives $z$. The distribution of $x$
does not depend on $y$, and the distribution of $z$ does not depend
on $s$. Connecting the resource channel to this box gives
\[
 \Theta_N[T](z\mid s)=\sum_{x,y}T(y\mid x)N(x,z\mid s,y).
\]
Non-signalling correlations include those obtainable from shared
randomness or shared entanglement. Their optimal performance therefore
gives a converse bound for both assistance classes
\cite{cubitt2011zero}. These correlations do not themselves transmit a
message: all communication passes through the resource channel.
Figure~\ref{fig:superchannel} shows the corresponding transformation.
The sender's resource input is chosen without access to the resource
output. No feedback is supplied. We measure error by the maximum,
over target input words, of the total variation distance between the
simulated and target output laws. This ensures that the approximation
holds for every target input, without assuming an input distribution.

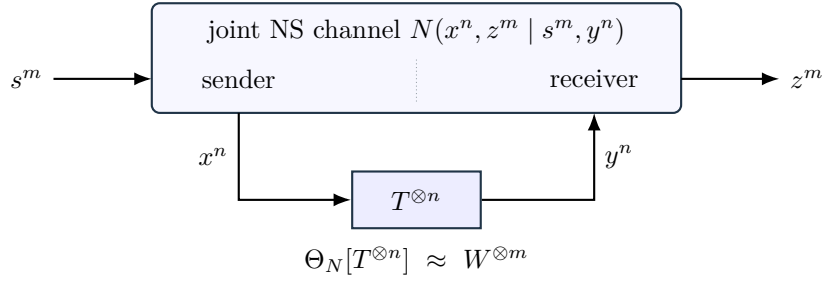
\begin{figure}[t]
\centering
\begin{tikzpicture}[>=Latex, font=\small,
 resource/.style={draw=ink,thick,minimum width=1.7cm,minimum height=.75cm,fill=blue!7}]
 \draw[draw=ink,thick,rounded corners=3pt,fill=blue!3]
       (0,0) rectangle (7,1.45);
 \node at (3.5,1.08) {joint NS channel $N(x^n,z^m\mid s^m,y^n)$};
 \node at (1.15,.45) {sender};
 \node at (5.85,.45) {receiver};
 \draw[densely dotted,ink!45] (3.5,.15) -- (3.5,.72);
 \draw[->,thick] (-1.3,.45) node[left] {$s^m$} -- (0,.45);
 \draw[->,thick] (7,.45) -- (8.3,.45) node[right] {$z^m$};
 \node[resource] (t) at (3.5,-1.15) {$T^{\otimes n}$};
 \draw[->,thick] (1.15,0) -- node[left] {$x^n$}
       (1.15,-1.15) -- (t.west);
 \draw[->,thick] (t.east) -- (5.85,-1.15)
       -- node[right] {$y^n$} (5.85,0);
 \node at (3.5,-1.95)
 {$\Theta_N[T^{\otimes n}]\ \approx\ W^{\otimes m}$};
\end{tikzpicture}
\caption{A non-signalling superchannel. The upper box is one bipartite
conditional distribution, not two local maps supplied with a shared
classical variable. Its sender marginal is independent of $y^n$, and
its receiver marginal is independent of $s^m$. Connecting its output
$x^n$ to its input $y^n$ through $T^{\otimes n}$ gives the simulated
channel. The return wire is not feedback to the sender.}
\label{fig:superchannel}
\end{figure}

The general result allows all dispersions in
$[V_{\min}(K),V_{\max}(K)]$, provided
$C_K>0$ and $V_{\min}(K)>0$ for $K=T,W$
(Theorem~\ref{thm:general}), but to keep the main text simple,
we assume (the general result and its proof are given in Appendix~\ref{app:nonunique})
\[
 V_{\min}(T)=V_{\max}(T)=V_T>0,\qquad
 V_{\min}(W)=V_{\max}(W)=V_W>0.
\]
This specialization includes nontrivial binary symmetric and binary
erasure channels and does not require a unique capacity-achieving input.
Let $m^*_{\NS}(n,\epsilon;T,W)$ be the largest target length attainable
with error at most $\epsilon$.

\begin{theorem}[Second-order channel interconversion]
\label{thm:main}
Under these assumptions, for every fixed $0<\epsilon<1$,
\begin{equation}
 m^*_{\NS}(n,\epsilon;T,W)
 =rn+\frac{d_\epsilon}{C_W}\sqrt n+O(\log n),\qquad
 r=\frac{C_T}{C_W},
 \label{eq:main-length}
\end{equation}
where the coefficient is the one-dimensional Gaussian optimization
\begin{equation}
 d_\epsilon=
 \inf_{\epsilon<\eta<1}
 \left\{\sqrt{V_T}\,\Phi^{-1}(\eta)
       -\sqrt{rV_W}\,\Phi^{-1}(\eta-\epsilon)\right\}.
 \label{eq:quantile}
\end{equation}
Here $\Phi$ is the standard normal distribution function.
\end{theorem}

Equivalently, the optimal rate is
$r+d_\epsilon/(C_W\sqrt n)+O(\log n/n)$. The capacities determine the
first-order conversion rate. The two total information variances
$nV_T$ and $rnV_W$ determine the correction.
The formula is a difference of Gaussian quantiles optimized over one
parameter, rather than the sum of the separate coding and simulation
penalties. Section~\ref{sec:second-order} proves the theorem.
Appendix~\ref{app:additive} also gives an exact characterization for
additive channels without a positive-dispersion assumption, including
their zero-dispersion cases. General channels with $V_{\min}=0$ are
not covered by the Gaussian theorem.

\subsection{Relation to state dichotomies}
The comparison of statistical experiments goes back to Blackwell,
Sherman, and Stein. Le Cam's deficiency theory treats approximate
comparison
\cite{blackwell1953equivalent,sherman1951theorem,lecam1964sufficiency}.
A classical state dichotomy is a binary experiment, specified by a pair
of distributions $(p,q)$. Its conversion to $(p',q')$ asks for a
stochastic map that approximates $p'$ from $p$ while mapping $q$ to $q'$
exactly. Renes' relative submajorization
theorem expresses its optimal error as a gap between testing curves
\cite{renes2016relative}. The corresponding second-order conversion
law is determined in \cite{lipkabartosik2024dichotomies}.
For channels, the corresponding reference objects are the replacer
channels $\mathcal R_Q(y\mid x)=Q(y)$, which carry no information
about their input. A single NS superchannel satisfies
\begin{equation}
 \begin{aligned}
  \Theta_N[T]&=\widetilde W\approx W,\\
  \Theta_N[\mathcal R_Q]&=\mathcal R_{D_NQ}
       &&\text{for every output law }Q,
 \end{aligned}
 \label{eq:intro-null-family}
\end{equation}
where $D_N(z\mid y)=\sum_xN(x,z\mid s,y)$ is a single stochastic
map fixed by the protocol. Thus the same transformation maps every
input-independent channel to another input-independent channel,
rather than preserving just one
chosen reference. The output reference $D_NQ$ is determined by the
protocol and is not prescribed independently. This also motivates the
choice of NS assistance: among classical superchannels,
these are precisely the transformations that preserve the family of
channels carrying no information \cite{takagi2020resource}.
Section~\ref{sec:setup} proves this directly from the two marginal
conditions.

For each target input law $P_S$, this channel transformation induces
a stochastic map between the joint input--output experiments
\[
 \bigl(P_X(x)T(y\mid x),P_X(x)Q(y)\bigr)
 \quad\longmapsto\quad
 \bigl(P_S(s)\widetilde W(z\mid s),P_S(s)(D_NQ)(z)\bigr),
\]
with $P_X$ induced by the sender marginal. The same induced map works
for every $Q$ (Proposition~\ref{prop:converse}). The channel problem therefore has two requirements. The approximation
must hold for every target input, and the same protocol must map every
replacer channel to another replacer channel. A transformation of one
chosen binary experiment need not satisfy either condition.

The distinction is operational. A state map can generate the target
input label itself and correlate it with the output. A channel protocol
must instead respond to an input supplied from outside. For example,
a map for one fixed pair can appear to create a noiseless bit from a
channel with constant output, while violating the receiver NS condition
(Section~\ref{sec:universal-lifting}). Our construction enforces the
channel constraints at logarithmic communication cost. This is the
sense in which the theory of state dichotomies extends to channels.
It does not give a one-shot equivalence with the conversion of a fixed
pair. The NS conditions must still be imposed.

\subsection{From binary experiments to universal channel protocols}
\label{sec:intro-lifting}
Our starting point is the comparison of testing curves in
\cite[Theorem~2]{renes2016relative}. Information spectrum estimates
and Gaussian limits then suggest the second-order profile
\cite{hayashi2009information,polyanskiy2010channel,
lipkabartosik2024dichotomies}. To prove a channel theorem, we need a
protocol attaining that profile and a converse valid for every allowed
protocol. Classical comparison theory already treats whole families
of distributions. The extra requirement here is that a transformation
of joint input--output distributions must respect the separation
between sender and receiver. In particular, it must respond to the
target input rather than generate that input itself.

\paragraph{Converse.}
This is more standard and the state map induced by a channel protocol is due to Oufkir, Yao, and
Berta \cite[Lemma~24]{oufkir2026interconversion}. It works for every
resource reference law, but the target reference is fixed by the
protocol and cannot be chosen independently. Resource inputs may also
be correlated. Appendix~\ref{app:gaussian} handles these choices using
uniform type-class estimates, bounds on output distributions, and
Markov's inequality. Near capacity-achieving inputs, we use the
quadratic decay of mutual information
\cite{caotomamichel2023quadratic}. A central limit theorem for one
fixed iid input would not justify optimization over all these choices.

\paragraph{Achievability.}
We use the method of types \cite{csiszar1998types}. First, restrict
the target inputs to one type class. Averaging over target coordinate
permutations makes the target input label uniform on that class.
This gives the sender marginal needed to turn the state map into a
conditional distribution for a channel protocol. The same symmetry
makes the error equal for every input in the class. Next, average independently over resource permutations. The receiver
marginal is then constant on each output type class. A reference
distribution that assigns equal mass to these classes bounds this
marginal at every output, up to a polynomial factor. We then enforce the receiver NS condition. The construction extends
the diagonal and off-diagonal completion used in NS simulation
\cite[Proposition~18]{cao2024channel}, while retaining the noisy
resource. When the auxiliary symbols agree, we use the desired
conditional distribution. When they differ, we add a nonnegative
correction. Both NS conditions then hold for every input to the box. A second auxiliary message identifies the target type. Its mismatched
branches must also be completed so that the full protocol remains NS.
Ordinary coding below capacity sends both auxiliary messages using
$O(\log n)$ of the prescribed resource uses, with error at most $n^{-2}$.
The total error is bounded by the largest typewise error plus the
auxiliary coding error. There is no factor equal to the number of
types. Section~\ref{sec:universal-lifting} gives the NS channels.
Proposition~\ref{prop:paid} accounts for their communication cost.

\subsection{Finite-error resonance and short blocklengths}
The limiting error at length $m=\lfloor rn+(d/C_W)\sqrt n\rfloor$ is
\begin{equation}
 \mathcal E(d)=\sup_{\xi\in\mathbb R}
 \left\{\Phi\!\left(\frac{\xi}{\sqrt{V_T}}\right)
       -\Phi\!\left(\frac{\xi-d}{\sqrt{rV_W}}\right)\right\}.
 \label{eq:profile}
\end{equation}
This is the sesquinormal profile $S_\nu(d/\sqrt{V_T})$ of
\cite[Lemma~1]{lipkabartosik2024dichotomies}, with $\nu=rV_W/V_T$.
Here we show that this profile is also attainable by universal channel
protocols.
The coefficient in Theorem~\ref{thm:main} is obtained by inverting
this limiting error. In particular, define the threshold error
\begin{equation}
 \epsilon_*:=\mathcal E(0)
 =\left|\Phi(x_\nu)-\Phi(x_\nu/\sqrt\nu)\right|,
 \qquad
 x_\nu=\sqrt{\frac{\nu\ln\nu}{\nu-1}},\quad
 \nu=\frac{rV_W}{V_T},
 \label{eq:threshold-error}
\end{equation}
with $\epsilon_*=0$ at $\nu=1$. If $\nu\ne1$, then
$0<\epsilon_*<1/2$ and
\begin{equation}
 d_{\epsilon_*}=0,\qquad
 \frac{m^*_{\NS}(n,\epsilon_*;T,W)}{n}
 =r+O\!\left(\frac{\log n}{n}\right).
 \label{eq:finite-error-resonance}
\end{equation}
Thus, at a fixed and possibly small error, the $n^{-1/2}$ rate
correction vanishes and only higher-order terms remain. This is the
threshold-error interpretation of finite-error resonance illustrated
for state dichotomies in Fig.~2 of
\cite{lipkabartosik2024dichotomies}. Such a threshold exists for every
pair with $\nu\ne1$. Its size depends on the difference between the two
ratios of dispersion to capacity. For $\epsilon<\epsilon_*$ the coefficient
is negative, and for $\epsilon>\epsilon_*$ it is positive.

The total information fluctuations match when
\begin{equation}
 \frac{V_T}{C_T}=\frac{V_W}{C_W}.
 \label{eq:resonance-condition}
\end{equation}
In this case,
\begin{equation}
 \mathcal E(d)=\left[2\Phi\!\left(\frac{d}{2\sqrt{V_T}}\right)-1\right]_+,
 \qquad
 d_\epsilon=2\sqrt{V_T}\,
                 \Phi^{-1}\!\left(\frac{1+\epsilon}{2}\right).
 \label{eq:resonance}
\end{equation}
Here the error at the capacity-ratio length tends to zero. This
zero-threshold resonance is not restricted to identical channels.
In this case, $d_\epsilon>0$ at every fixed $\epsilon>0$, so
zero-threshold resonance is distinct from cancellation at a positive
threshold error.

\paragraph{Non-signalling resonance.}
Fix $W=\mathrm{BSC}(q)$ with $q=0.01$ and take
$T=\mathrm{BEC}(\delta)$, where $\delta$ is the \emph{erasure}
probability. With $C_q=1-h(q)$ and
$V_q=q(1-q)\log^2((1-q)/q)$, resonance occurs at
\[
 \delta_*=\frac{V_q}{C_q}\simeq0.473,
 \qquad r=\frac{1-\delta_*}{C_q}\simeq0.573.
\]
Figure~\ref{fig:resonance}(a) plots the NS threshold as $\delta$
varies. The two $1\%$ intersections are approximately
$\delta=0.436$ and $\delta=0.514$: at either endpoint
the second-order term vanishes at $1\%$ error, while it is positive
between them. The SR threshold curve is not known for this family, unlike the
packet example below. The stronger common-encoder converse in
Appendix~\ref{app:sr-converse} proves that SR has strictly positive
limiting error at the NS zero $\delta_*$, and a strictly smaller
second-order rate at a sufficiently small fixed error
(Proposition~\ref{app-sr:gap}).

The asymptotic remainder alone does not guarantee that the capacity
ratio is attainable at a short blocklength. For additive channels, we
can check this using an exact finite-blocklength formula. For $T=\mathrm{BSC}(0.1)$ and
$W=\mathrm{BSC}(0.05)$,
\[
 r\simeq0.744,\qquad \epsilon_*\simeq0.042,
 \qquad e_{\NS}(25,\lfloor25r\rfloor)
 =e_{\NS}(25,18)<0.036<\epsilon_*.
\]
Thus the capacity-ratio target length, rounded down, is attainable
already with $25$ resource uses at the error where the second-order
term is exactly zero. Figure~\ref{fig:finite-bsc} compares the
finite-blocklength rates with the second-order predictions.
Appendix~\ref{app:finite-resonance} gives the exact certificate,
retaining all higher-order and lattice effects, and an explicit
Gaussian upper bound. It also certifies
$e_{\NS}(20,81;\mathrm{BSC}(0.4),\mathrm{BSC}(0.45))<0.008$:
at $1\%$ error this pair attains a rate $81/20=4.05$ above its
capacity ratio $r\simeq4.020$. In the latter example the error
exceeds $\epsilon_*\simeq0.002$, so the second-order coefficient
is positive, rather than zero.

\paragraph{Shared-randomness resonance.}
Let $E_2(\delta)$ transmit a two-bit packet intact with probability
$1-\delta$, and otherwise erase both bits. Fix the target
$W=\mathrm{BEC}(1/2)$ and take $T=E_2(\delta)$, $0<\delta<1/2$.
These are fixed memoryless channels. One use of the resource is not
equivalent to two independent BEC uses. Their information parameters are
\[
 C_T=2(1-\delta),\quad V_T=4\delta(1-\delta),\quad
 C_W=\tfrac12,\quad V_W=\tfrac14,
 \qquad r=4(1-\delta),\quad\nu=\frac1{4\delta}.
\]
Proposition~\ref{app-sr:packet} proves that SR, EA, and NS have the
same limiting error profile throughout this family. Consequently
Figure~\ref{fig:resonance}(b) is a shared-randomness analogue of the
threshold-error plot in Fig.~2 of \cite{lipkabartosik2024dichotomies}.
It has an interior zero at $\delta=1/4$, where $r=3$ and
$V_T=rV_W=3/4$. At the two values
\[
 \delta\simeq0.230,\qquad \delta\simeq0.272,
\]
the threshold is $1\%$: the second-order coefficient at this error is
exactly zero for two distinct channels. Between these values it is
positive. At $\delta=1/4$ it is the \emph{threshold error}, not the
coefficient at fixed positive error, that vanishes.

The SR construction proceeds by shared random placement of the
transmitted packet bits among the target positions. The receiver
observes the number of surviving bits and selects a subset of their
positions to match the target's visible-bit count. At resonance, the
two counts have matching means and variances, and the coupling failure
probability tends to zero. This uses neither feedback nor NS-to-SR rounding.
It also gives finite guarantees: at $\delta=1/4$, $250$ packet uses
simulate $750$ target uses with error below $9.719\times10^{-3}$. Already
$16$ packet uses simulate $47$ target uses, just one below the
capacity-ratio length $48$, with error below $9.971\times10^{-3}$.
These statements follow from exact binomial arithmetic, not a
prescribed higher-order term (Appendix~\ref{app-sr:packet-resonance}).

\begin{figure}[!htbp]
\centering
\setlength{\abovecaptionskip}{4pt}
\setlength{\belowcaptionskip}{0pt}
\includegraphics[width=\linewidth]{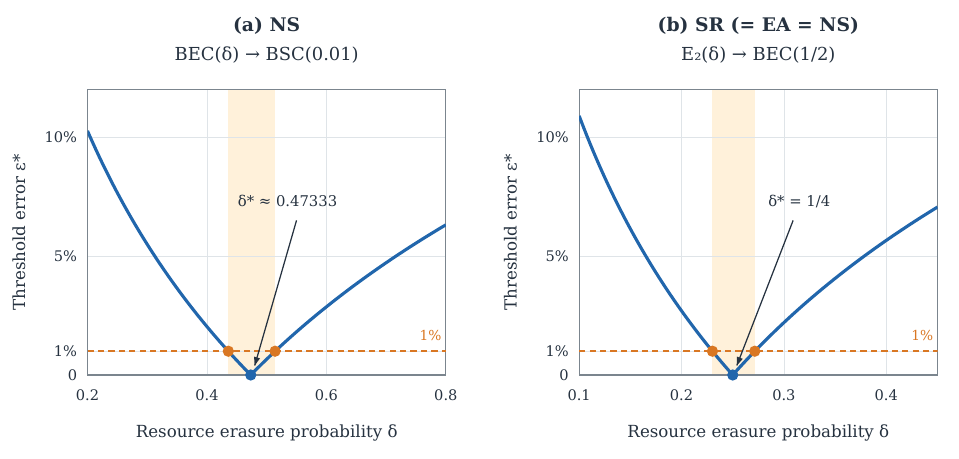}
\caption{Threshold errors at the capacity ratio, analogous to Fig.~2
of \cite{lipkabartosik2024dichotomies}.
(a) NS conversion from a BEC to a BSC. The SR curve is not known.
(b) SR conversion from a two-bit packet-erasure channel to a BEC.
Here the SR, EA, and NS limits coincide.
At the orange points, the second-order coefficient at $1\%$ error
vanishes. It is nonnegative in the shaded intervals.
At the blue zeros, the error at the capacity ratio tends to zero.
These are limiting curves, not finite-blocklength error bounds.}
\label{fig:resonance}

\vspace{1mm}
\includegraphics[width=\linewidth]{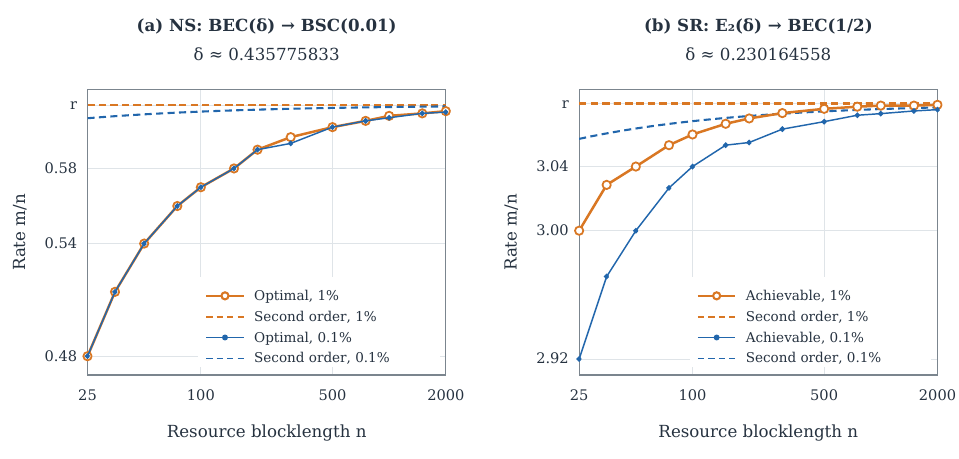}
\caption{Rates at the smaller $1\%$ threshold roots of
Figure~\ref{fig:resonance}.
Solid curves evaluate exact finite-blocklength formulas, giving
optimal NS rates in (a) and achievable SR rates in (b).
Dashed curves show $r+b_\epsilon/\sqrt n$, with $b_{1\%}=0$.
They are not finite-blocklength bounds.
The two finite curves in (a) often coincide because lengths and
Hamming weights are integers.}
\label{fig:resonance-rates}
\end{figure}

Figure~\ref{fig:resonance-rates} shows the corresponding rates as a
function of blocklength for both examples. We fix the smaller $1\%$
threshold root in each family and compare errors $1\%$ and $0.1\%$.
Writing $b_\epsilon=d_\epsilon/C_W$, the second-order rate is
$r+b_\epsilon/\sqrt n$: at $1\%$ it is the constant $r$, whereas
at $0.1\%$ the coefficients are approximately $-0.110$ for the
SR packet example and $-0.035$ for the NS BEC-to-BSC example.
The finite curves retain all higher-order and integer effects.
They are protocol-achievable rates in the SR panel and optimal rates
in the NS panel, computed from the exact formulas in
Appendices~\ref{app-sr:packet-resonance} and~\ref{app:bec-bsc-exact}.
In particular, cancellation of the second-order term need not remove
the visible finite-blocklength deficit for the near-identity target.

\subsection{Shared randomness and related work}
The NS characterization gives a converse bound for shared-randomness
assistance, but the NS construction need not be implementable with
shared randomness.
Appendix~\ref{app:sr} establishes agreement of the second-order
coefficients for a direction of erasure-channel conversion, including
the nonidentical packet-to-bit resonance in Figure~\ref{fig:resonance}(b),
and for BSC self-conversion. It also proves a strict separation for
the dispersion-matched BEC-to-BSC pair in panel (a).
Shared randomness can therefore exhibit resonance, but
\eqref{eq:resonance-condition} is not a sufficient condition in general.
In particular, panel (a) is not a shared-randomness formula: the full
SR threshold curve for BEC-to-BSC conversion remains open.

\paragraph{Related work and proof ingredients.}
At the noiseless endpoints, the relevant results are the
finite-blocklength coding theorems of Hayashi, and Polyanskiy, Poor, and Verd\'u
\cite{hayashi2009information,polyanskiy2010channel}, its NS
interpretation due to Matthews \cite{matthews2012linear}, and the
channel simulation results of Cao et al.\ \cite{cao2024channel}.
Our testing comparison uses Renes' relative submajorization theorem
\cite{renes2016relative}. Its Gaussian profile is the sesquinormal
profile derived in \cite{lipkabartosik2024dichotomies}.
The induced state map in the converse is due to Oufkir, Yao, and Berta
\cite[Lemma~24]{oufkir2026interconversion}. The marginal completion
extends the NS simulation construction of Cao et al.\
\cite[Proposition~18]{cao2024channel}. The additional step is to implement these ingredients uniformly over
target types, charge all communication to the prescribed resource, and
allow every reference law in the converse.
The strong converse exponent of channel interconversion
\cite{oufkir2026interconversion} concerns fixed rates above capacity,
where success decays exponentially. It does not determine the
fixed-error correction studied here.

At the noiseless endpoints, NS relaxations can already be rounded
with a sub-second-order communication cost: for coding by
Barman and Fawzi \cite{barman2018algorithmic}, and for simulation by
rejection sampling \cite{cao2024channel,oufkir2026optimality}.
These results preserve the fixed-error second-order terms at the
respective endpoints. They do not, however, provide a rounding of a
general conversion protocol that retains a noisy resource. Such a
rounding would address the corresponding obstacle for SR
interconversion.

There is also a close connection with joint source--channel coding
\cite{kostina2013lossy,yaguchi2019second,kochman2020distortion}.
In both problems, concatenation through noiseless communication can
be suboptimal at second order. Channel conversion additionally requires reproducing the conditional
output distribution uniformly over inputs, rather than satisfying a
single decoding or distortion criterion.
Source--channel constructions suggest achievable bounds. The
binary-source broadcast converse of Kochman, Ordentlich, and Polyanskiy
\cite{kochman2020distortion} supplies the idea of introducing an
auxiliary variable for the stronger SR converse in
Appendix~\ref{app:sr-converse}. We also condition on success and
estimate the probabilities of intervals of counts using Gaussian limits.
A matching general characterization has not been established.

\subsection{Conclusion and further questions}
\label{sec:conclusion}
We characterize the second-order interconversion of finite classical
channels with positive capacities and positive minimum
capacity-achieving dispersions under non-signalling superchannels.
The optimal coefficient is governed by Gaussian information
fluctuations, with a mixture over resource dispersions and a worst-case
target variance in the general case (Theorem~\ref{thm:general}).
When each channel has equal minimum and maximum dispersions, the
characterization
reduces to the sesquinormal expression for state dichotomies.
The channel extension requires uniformity over target inputs and
both non-signalling marginal constraints. Our explicit construction enforces these constraints at logarithmic
communication cost.
For additive channels, the exact finite-blocklength characterization
also covers zero-dispersion cases (Appendix~\ref{app:additive}).

The corresponding problem under shared randomness remains open in
general. Appendix~\ref{app:sr} establishes both agreement with the NS
coefficient and a strict separation, without determining the separated
SR coefficient. Packet-to-bit erasure conversion gives an SR resonance between
distinct channels: random placement and local thinning couple the
numbers of available and required bits.
The BEC-to-BSC separation shows why matching the total information
variances at the capacity-ratio length does not give a universal SR
resonance criterion.
Joint source--channel coding suggests constructions
that preserve resource noise, while broadcast converses impose
compatibility across different noise realizations. Determining whether
these approaches yield a matching characterization is an open problem.
Poisson-matching constructions \cite{li2021poisson} offer a possible
alternative to separate coding and simulation, but would additionally
have to reproduce the conditional output distribution uniformly over
inputs.

Entanglement assistance lies between shared randomness and NS
assistance. Its coefficient agrees with both whenever their coefficients coincide,
but remains undetermined for the separated BEC-to-BSC example.
Extending the common-encoder converse to quantum decoders would
require a separate argument. Further questions concern general
nonadditive zero-dispersion channels, including intervals with
$V_{\min}=0<V_{\max}$, and third-order refinements of the
positive-dispersion expansion.

\subsection*{Acknowledgements}
\addcontentsline{toc}{subsection}{Acknowledgements}
We acknowledge discussions with
Marco Tomamichel and Rabsan Galib Ahmed. MB acknowledges support from the European Research Council
(ERC Grant Agreement No.~948139) and the Excellence Cluster Matter and
Light for Quantum Computing (ML4Q-2). This work was supported in part
by a grant of access to OpenAI models through the ChatGPT for Academic
Researchers program.

\clearpage
\tableofcontents
\clearpage

\section{Channel interconversion and information parameters}
\label{sec:setup}
All alphabets are finite. Logarithms are base two, except $\ln$.
Write $\mathcal P(\mathcal A)$ for probability laws on $\mathcal A$,
$\TV(P,R)=\frac12\sum_a|P(a)-R(a)|$, and $[u]_+=\max(u,0)$.
The induced output law is $(PK)(b)=\sum_aP(a)K(b\mid a)$.
When $PK$ denotes a \emph{joint} law below, this will be stated explicitly.
Its value is $P(a)K(b\mid a)$. A stochastic matrix acts on column
probability vectors and has column sums one.
Let $\Phi,\phi$ denote the standard normal distribution function and density.

\subsection{Operational task and non-signalling constraints}
An NS protocol is described by a bipartite channel $N(x,z\mid s,y)$ satisfying
\begin{equation}
 \sum_zN(x,z\mid s,y)=E(x\mid s),\qquad
 \sum_xN(x,z\mid s,y)=D(z\mid y).
 \label{eq:ns-marginals}
\end{equation}
Connecting the resource channel $K$ gives the simulated channel
\begin{equation}
 (\Theta_N[K])(z\mid s)=\sum_{x,y}K(y\mid x)N(x,z\mid s,y).
 \label{eq:contraction}
\end{equation}
The first constraint ensures normalization for every $K$.
The second shows that, for every replacer
$\mathcal R_R(y\mid x)=R(y)$,
\begin{equation}
 \Theta_N[\mathcal R_R](z\mid s)=\sum_yR(y)D(z\mid y),
 \label{eq:replacer}
\end{equation}
which is independent of $s$.

Conversely, suppose a nonnegative array in
\eqref{eq:contraction} gives a valid channel for every inserted channel
and preserves every replacer. These properties imply the NS constraints.
See also
\cite{takagi2020resource}. Indeed, set
$A_s(x,y)=\sum_zN(x,z\mid s,y)$. Normalization requires
$\sum_{x,y}K(y\mid x)A_s(x,y)=1$ for every $K$.
Varying one column $K(\cdot\mid x)$ while holding the others fixed
shows that $A_s(x,y)$ is independent of $y$, giving the sender
marginal. For each fixed $y$, preservation of the point-mass
replacer $\mathcal R_{\delta_y}$ says that
$\sum_xN(x,z\mid s,y)$ is independent of $s$, giving the receiver
marginal. Thus, among valid transformations of this form, preserving
the entire family of replacers characterizes NS transformations.
Preserving one chosen reference channel is not enough.

Let $T:\mathcal X\to\mathcal Y$ be the resource channel and
$W:\mathcal S\to\mathcal Z$ the target. Using block alphabets in
\eqref{eq:ns-marginals}, define
\begin{align}
 e_{\NS}(n,m;T,W)
 &=\min_N\max_{s^m}
   \TV\bigl(\Theta_N[T^{\otimes n}](\cdot\mid s^m),
                         W^{\otimes m}(\cdot\mid s^m)\bigr),
 \label{eq:error}\\
 m^*_{\NS}(n,\epsilon;T,W)&=\max\{m:e_{\NS}(n,m;T,W)\le\epsilon\}.
 \label{eq:length}
\end{align}
At each fixed pair of blocklengths the NS feasible set is a compact
polytope, so the error minimum is attained. Under the hypotheses of
Theorem~\ref{thm:main}, the length maximum is finite for large $n$.
The bounds below also recover the first-order theorem.
We omit $T,W$ when they are fixed. Error is nondecreasing in $m$:
fixing unused target input coordinates and discarding the corresponding
outputs gives a protocol for any shorter target block.
Equivalently, the error is the largest joint-law TV distance over all
input laws $P_{S^m}$. The criterion therefore controls every conditional
output distribution, not just an average over a prescribed source.

An SR protocol has
$N(x,z\mid s,y)=\sum_\rho p(\rho)E_\rho(x\mid s)D_\rho(z\mid y)$.
EA protocols use local measurements on a preshared quantum state.
For $A\in\{\SR,\EA\}$, define $e_A$ by \eqref{eq:error} with an
infimum over the corresponding assistance class, and $m_A^*$ by
\eqref{eq:length}. The infimum convention does not assume that an
optimal unrestricted EA protocol is attained.
No feedback or additional communication is allowed, and
\[
 \SR\subseteq\EA\subseteq\NS,\qquad
 e_{\NS}\le e_{\EA}\le e_{\SR}.
\]
The resource input coordinates may be correlated. Every auxiliary
channel used in a proof will be implemented within the $n$ resource uses.

\subsection{Capacity and a single dispersion}
For $A\sim P$ and $B\sim K(\cdot\mid A)$ define
\[
 i_{P,K}(a;b)=\log\frac{K(b\mid a)}{(PK)(b)},\qquad
 I(P,K)=\E i_{P,K}(A;B),\qquad
 V(P,K)=\E\Var(i_{P,K}(A;B)\mid A).
\]
Define the capacity and the set of capacity-achieving input distributions by
$C_K=\max_P I(P,K)$ and
$\Pi_K=\{P:I(P,K)=C_K\}$. All $P\in\Pi_K$ induce the same output
law $P_B^*(K)$, and
$D(K(\cdot\mid a)\Vert P_B^*(K))=C_K$ for $a$ in their support.
Consequently, on $\Pi_K$ the conditional variance equals the full
variance, and
\begin{equation}
 V(P,K)=\sum_{a,b}P(a)K(b\mid a)
       \left(\log\frac{K(b\mid a)}{P_B^*(K)(b)}-C_K\right)^2.
 \label{eq:dispersion}
\end{equation}
This expression is affine in $P$ on $\Pi_K$.
In the main text we assume that it has one value $V_K>0$ and that
$C_K>0$. This assumption allows multiple capacity-achieving inputs,
provided their dispersions agree. The interval-valued case is treated
in Appendix~\ref{app:nonunique}.

We use $\mathrm{BSC}(q)$ for crossover probability $q$ and
$\mathrm{BEC}(\delta)$ for erasure probability $\delta$.
Write
\[
 h(x)=-x\log x-(1-x)\log(1-x),\quad
 C_x=1-h(x),\quad \ell_x=\log\frac{1-x}{x},\quad V_x=x(1-x)\ell_x^2.
\]
For $0<q<1/2$ and $0<\delta<1$,
\begin{center}
\begin{tabular}{lll}
\toprule
Channel & Capacity & Dispersion\\
\midrule
$\mathrm{BSC}(q)$ & $C_q$ & $V_q$\\
$\mathrm{BEC}(\delta)$ & $1-\delta$ & $\delta(1-\delta)$\\
\bottomrule
\end{tabular}
\end{center}
Thus these basic channels satisfy the singleton-dispersion assumption
throughout the stated parameter ranges. Noiseless channels have zero
dispersion and require separate endpoint results.

A word $a^k$ has empirical law
$P_{a^k}(a)=k^{-1}|\{i:a_i=a\}|$, called its type.
We use the standard type-class estimates of \cite{csiszar1998types}.
Write $\mathcal P_k(\mathcal A)$ for the set of types,
$\mathcal T_P^k$ for a type class, and $U_{\mathcal T_P^k}$ for its
uniform law. There are at most $(k+1)^{|\mathcal A|}$ types.
Every permutation-invariant law is a mixture of these uniform laws.
We use the following universal output reference:
\begin{equation}
 R_{\mathcal B,k}^{\mathrm u}
 =\frac1{|\mathcal P_k(\mathcal B)|}
       \sum_{P\in\mathcal P_k(\mathcal B)}U_{\mathcal T_P^k}.
 \label{eq:universal-reference}
\end{equation}
In particular $R^{\otimes k}\le
|\mathcal P_k(\mathcal B)|R_{\mathcal B,k}^{\mathrm u}$ pointwise
for every one-letter output law $R$.

\section{Binary experiments and non-signalling interconversion}
\label{sec:oneshot}
We first derive a converse from testing curves and then construct
non-signalling protocols from transformations of binary experiments.
These finite-blocklength statements require no dispersion assumption.

\subsection{Converse: the induced dichotomy}
For probability vectors $p,q$, define the testing curve
\[
 \alpha_y(p,q)=\max\{t\cdot p:0\le t\le1,\ t\cdot q\le y\},
 \qquad 0\le y\le1,
\]
and define
$\delta(p,q;p',q')=\min_{M:\,Mq=q'}\TV(Mp,p')$, where $M$ is
stochastic. The comparison theorem for binary experiments is
\begin{equation}
 \delta(p,q;p',q')
 =\sup_{0\le y\le1}\{\alpha_y(p',q')-\alpha_y(p,q)\}.
 \label{eq:testing-gap}
\end{equation}
This is the normalized, exact-reference form of the classical
randomization criterion \cite{torgersen1991comparison}. See also
\cite[Theorem~2 and Proposition~6]{renes2016relative}.
We recall a finite-alphabet proof for completeness. Pulling tests back
through $M$ gives the lower bound on the left side of
\eqref{eq:testing-gap}. For the reverse bound, minimax expresses that side as
$\max_{0\le f\le1}\{p'\cdot f-\max_{M:Mq=q'}(Mp)\cdot f\}$.
Order the source atoms by $p/q$ and the target atoms by $f$.
By splitting the reference mass, a monotone transport maximizes the
assigned $p$-mass simultaneously for every upper level set of $f$.
Integrating over these level sets bounds the expression by the right
side of \eqref{eq:testing-gap}. Atoms with $q=0$ can be assigned to
a target atom where $f$ is maximal.

To pass from testing curves to Gaussian limits, we use the following
elementary information-spectrum bounds. For $j=\log(p/q)$, $h\ge0$
and $\gamma>0$,
\begin{equation}
 \Pr_p\{j\ge h\}\le\alpha_{2^{-h}}(p,q)
 \le\Pr_p\{j\ge h-\gamma\}+2^{-\gamma}.
 \label{eq:spectrum}
\end{equation}
The threshold test in the left inequality has $q$-mass at most $2^{-h}$.
For the right inequality, the contribution of any feasible test on
$j<h-\gamma$ has $p$-mass at most $2^{-\gamma}$.

\begin{proposition}[One-shot converse]
\label{prop:converse}
Let $N$ convert $T$ to a channel $\widetilde W$ whose worst-row error
from $W$ is at most $\epsilon$. Choose any target input law $P_S$ and any
resource output reference $R_Y$. Set
\[
 P_X(x)=\sum_sP_S(s)E(x\mid s),\qquad
 R_Z(z)=\sum_yR_Y(y)D(z\mid y).
\]
With joint signal laws $P_{XY}=P_XT$ and $P_{SZ}=P_SW$,
\begin{equation}
 \epsilon\ge\delta(P_{XY},P_X\otimes R_Y;
                         P_{SZ},P_S\otimes R_Z).
 \label{eq:one-shot-converse}
\end{equation}
The statement applies unchanged to block channels.
\end{proposition}
\begin{proof}
We use the induced stochastic map of
\cite[Lemma~24]{oufkir2026interconversion}.
On columns with $P_X(x)>0$, put
$M_{(s,z),(x,y)}=P_S(s)N(x,z\mid s,y)/P_X(x)$.
The map depends on $N$ and $P_S$, but not on $R_Y$.
The first NS constraint makes $M$ stochastic, and the second gives,
simultaneously for every $R_Y$,
\[
 M(P_X\otimes R_Y)=P_S\otimes R_Z,\qquad
 MP_{XY}=P_S\widetilde W.
\]
Columns corresponding to zero-probability inputs can be defined arbitrarily.
The last joint law is within $\epsilon$ of $P_SW$ in TV.
\end{proof}
The reference $R_Z$ is induced by the converter. It cannot be fixed
to a convenient product law in the converse.

\subsection{Construction of non-signalling protocols}
\label{sec:universal-lifting}
We now construct a channel protocol from a state comparison.
Two additional requirements must be met: the state map must give an
NS channel, and its error must be controlled uniformly over target inputs.

\paragraph{One reference pair is insufficient.}
Consider a resource with a single input and deterministic output
$Y=0$, with output alphabet $\{0,1\}$. Choose the reference $U_Y$,
and let the target be the noiseless binary channel with uniform input
$U_S$. The stochastic map
\[
 M(s,z\mid y)=\tfrac12\mathbf1_{\{z=s\oplus y\}}
\]
satisfies
\[
 M\delta_0=U_S(s)\mathbf1_{\{z=s\}},\qquad
 MU_Y=U_S\otimes U_Z.
\]
Thus the two selected state dichotomies convert exactly. Nevertheless,
the resource is a replacer, and every NS image of it is a replacer.
Its worst-row error from the noiseless bit is at least $1/2$.
The conditional distribution $H(z\mid s,y)=2M(s,z\mid y)$ shows why:
its receiver marginal depends on $s$. Preserving the single averaged
reference $U_Y$ concealed this dependence. The state map generated
the label $S$ jointly with $Z$. It did not implement a channel for an
externally chosen $S$.

More generally, suppose a stochastic map $M$ has
$\sum_zM(s,z\mid x,y)=\pi(s)>0$, and define
\begin{equation}
 H(x,z\mid s,y)=\frac{P_X(x)M(s,z\mid x,y)}{\pi(s)}.
 \label{eq:reverse-state-map}
\end{equation}
It has sender marginal $P_X$. Its receiver marginal is independent of
$s$ if and only if there is one stochastic map $D$ such that
\begin{equation}
 M(P_X\otimes R_Y)=\pi\otimes(DR_Y)
 \quad\text{for every }R_Y.
 \label{eq:all-reference-lifting}
\end{equation}
Indeed, taking $R_Y=\delta_y$ tests the receiver marginal at each $y$.
The converse follows by linearity. A single reference-preservation
identity is strictly weaker. The completion below obtains this full
marginal identity from a pointwise domination bound.

Controlling average error is a separate issue. For any fixed input law,
\[
 \TV(P_S\widetilde W,P_SW)
 =\sum_sP_S(s)\TV(\widetilde W(\cdot\mid s),W(\cdot\mid s)).
\]
For example, the constant-zero channel is within $\eta$ of the
noiseless bit under the prior $P_S(1)=\eta$, but its worst-row error
is one. We will use type-class symmetry and a message identifying the
input type to obtain a uniform bound without multiplying the error by
the number of inputs.

\paragraph{Completing the receiver marginal.}
The following lemma extends the diagonal/off-diagonal completion
underlying the NS simulation linear program
\cite[Proposition~18]{cao2024channel}.
Here the construction retains the noisy resource while enforcing
both non-signalling marginals.

\begin{lemma}[Completion]
\label{lem:completion}
Suppose $H(x,z\mid s,y)\ge0$ has sender marginal
$\sum_zH(x,z\mid s,y)=P_X(x)$.
Let $B_s(z\mid y)=\sum_xH(x,z\mid s,y)$ and suppose
$B_s(z\mid y)\le cR_Z(z)$ for distributions $P_X,R_Z$
and an integer $c\ge2$.
Then $T\otimes\id_{[c]}$ can simulate the channel
$\sum_{x,y}T(y\mid x)H(x,z\mid s,y)$ with NS assistance.
\end{lemma}
\begin{proof}
For auxiliary input $i$ and output $j$ in $[c]$, define
\begin{equation}
 N(x,i,z\mid s,y,j)=
 \begin{cases}
 H(x,z\mid s,y)/c,&i=j,\\[1mm]
 \displaystyle\frac{P_X(x)[cR_Z(z)-B_s(z\mid y)]}{c(c-1)},&i\ne j.
 \end{cases}
 \label{eq:completion}
\end{equation}
The domination assumption makes every entry nonnegative.
Since $\sum_zB_s(z\mid y)=1$, summing over $z$
gives $P_X(x)/c$ on both diagonal and off-diagonal branches.
For the other marginal, sum over both sender outputs $x,i$:
\[
 \sum_{x,i}N(x,i,z\mid s,y,j)
 =\frac{B_s(z\mid y)}c+
   \frac{cR_Z(z)-B_s(z\mid y)}c=R_Z(z).
\]
Both identities hold on the entire enlarged input alphabet, including
mismatched auxiliary symbols. Hence $N$ is NS. Connecting the identity
channel leaves only the $c$ diagonal terms, whose sum gives exactly
the desired channel.
\end{proof}

The diagonal entries retain the statistical transformation.
The off-diagonal entries correct its marginals without affecting its
action on the ideal auxiliary channel. Defining only the diagonal
entries would not suffice: an NS box must be valid for all of its
inputs before it is connected to a particular channel.
This construction need not define an SR protocol. The auxiliary channel
will be implemented using the prescribed resource channel.

\paragraph{Obtaining the pointwise bound from a state converter.}
Fix a target input type $\tau$ of length $m$ and a
permutation-invariant resource input law $P_{X^k}$.
Put $U_\tau=U_{\mathcal T_\tau^m}$ and consider the dichotomies
\begin{align}
 &(P_{X^k}T^{\otimes k},\
                  P_{X^k}\otimes R_{\mathcal Y,k}^{\mathrm u}),\nonumber\\
 &(U_\tau W^{\otimes m},\
                  U_\tau\otimes(\tau W)^{\otimes m}).
 \label{eq:type-pairs}
\end{align}
Here the first member of each pair denotes its joint input--output law.
Choose an optimal stochastic map $M$ between these experiments. Its
target alphabet is $\mathcal T_\tau^m\times\mathcal Z^m$ on every source
column, including zero-probability columns.
Average $M$ independently over resource and target coordinate
permutations, acting simultaneously on the input and output word of
each channel. The signal and reference laws are invariant under these
permutations. Convexity of TV therefore preserves the error bound, and
the reference is still mapped exactly. The two averages serve different
purposes.

Write $x=x^k,y=y^k,s=s^m,z=z^m$ temporarily.
For each fixed $(x,y)$, transitivity on the target type class makes
$\sum_zM_{(s,z),(x,y)}$ constant in $s$ on $\mathcal T_\tau^m$.
Stochasticity therefore implies
$\sum_zM_{(s,z),(x,y)}=U_\tau(s)$.
Consequently
\[
 H(x,z\mid s,y)=\frac{P_{X^k}(x)M_{(s,z),(x,y)}}{U_\tau(s)}
\]
has sender marginal $P_{X^k}$.
Resource permutation invariance makes $B_s(z\mid y)$ constant on
each resource output type class. Let its value on the $j$th class
be $B_s^{(j)}(z)$ and let $J_k=|\mathcal P_k(\mathcal Y)|$.
The universal law $R_{\mathcal Y,k}^{\mathrm u}$ assigns each output
type class weight $1/J_k$. Exact reference preservation consequently
yields, separately for every $s$,
\begin{equation}
 (\tau W)^{\otimes m}(z)=\frac1{J_k}\sum_{j=1}^{J_k}B_s^{(j)}(z),
 \qquad B_s(z\mid y)\le J_k(\tau W)^{\otimes m}(z).
 \label{eq:type-domination}
\end{equation}
Lemma~\ref{lem:completion} applies with $c_k=\max(2,J_k)$.
The same target symmetry makes all row errors equal, so joint-law TV
is precisely the worst-row error on this class. The polynomial factor
$J_k$ appears in the auxiliary alphabet, not as a multiplier of the
approximation error. This gives a state conversion that can be completed
to a channel protocol while preserving its error bound.

\paragraph{One protocol for all target inputs.}
The resource law may depend on $\tau$. To combine the protocols for
different types while preserving the NS constraints, send the type index
through an additional identity channel of size $h_m=|\mathcal P_m(\mathcal S)|$.
Simply choosing the box indexed by $\tau(s)$ would not work: its
receiver marginal could then be $R_{\tau(s)}$, which depends on the
sender's input. The receiver marginal must instead be indexed by the
received header symbol, even when that symbol is incorrect.
If $N_\tau$ has marginals $P_\tau(x)/c_k$ and $R_\tau(z)$, define
\begin{equation}
 \widehat N(x,i,a,z\mid s,y,j,b)
 =\mathbf1_{\{a=\tau(s)\}}
 \begin{cases}
 N_{\tau(s)}(x,i,z\mid s,y,j),&b=\tau(s),\\
 P_{\tau(s)}(x)R_b(z)/c_k,&b\ne\tau(s).
 \end{cases}
 \label{eq:type-selection}
\end{equation}
Its sender marginal is
$\mathbf1_{\{a=\tau(s)\}}P_{\tau(s)}(x)/c_k$, and its receiver
marginal is $R_b(z)$. Thus the NS constraints hold even on mismatched
header symbols. The combined auxiliary alphabet has size
\begin{equation}
 L_{k,m}=c_kh_m\le2(k+1)^{|\mathcal Y|}(m+1)^{|\mathcal S|}.
 \label{eq:aux-size}
\end{equation}
The combined protocol has error equal to the largest typewise error,
without a factor equal to the number of types.

\subsection{Accounting for auxiliary communication}
\begin{proposition}[One-shot achievability with auxiliary cost]
\label{prop:paid}
If $T^{\otimes\ell}$ has a code for $L_{k,m}$ messages with maximal
error $\eta_\ell$, then
\begin{equation}
 e_{\NS}(k+\ell,m)
 \le\eta_\ell+\max_{\tau\in\mathcal P_m(\mathcal S)}
       \inf_{P_{X^k}\ {\rm permutation\ invariant}}
       \delta\bigl(\text{the two pairs in \eqref{eq:type-pairs}}\bigr).
 \label{eq:paid}
\end{equation}
If $C_T>0$ and $m=O(k)$, one can take
$\ell=O(\log k)$ and $\eta_\ell\le k^{-2}$.
\end{proposition}
\begin{proof}
Combine the two auxiliary inputs and outputs into single symbols
$i,j\in[L_{k,m}]$. Compose these ports of $\widehat N$ with the
encoder $E_{\rm c}(v\mid i)$ and decoder $D_{\rm c}(j\mid w)$ of the
ordinary channel code:
\[
 \overline N(x,v,z\mid s,y,w)
 =\sum_{i,j}E_{\rm c}(v\mid i)D_{\rm c}(j\mid w)
                         \widehat N(x,i,z\mid s,y,j).
\]
Summing the composed conditional distribution over either party's outputs verifies that
both NS marginal constraints are preserved.
For any NS channel and any two resource channels on the same alphabets,
\begin{equation}
 \max_s\TV(\Theta_N[K](\cdot\mid s),\Theta_N[K'](\cdot\mid s))
 \le\max_x\TV(K(\cdot\mid x),K'(\cdot\mid x)),
 \label{eq:contractivity}
\end{equation}
because the triangle inequality followed by the sender marginal
bounds the left side by
$\max_s\sum_xE(x\mid s)\TV(K(\cdot\mid x),K'(\cdot\mid x))$.
Replacing the auxiliary identity channel by the code therefore adds
at most $\eta_\ell$ to the error.

Finally, ordinary random coding at any fixed rate below $C_T$
has a strictly positive error exponent \cite{polyanskiy2010channel}.
Expurgating half the messages gives maximal error at most
$2^{1-c\ell}$ for a fixed $c>0$. By \eqref{eq:aux-size},
$2L_{k,m}$ messages and error at most $k^{-2}$ can both be achieved with
$\ell=A\log k$ uses for a sufficiently large constant $A$.
\end{proof}

\section{Second-order characterization}
\label{sec:second-order}
We prove Theorem~\ref{thm:main} by establishing matching converse and
achievability bounds. Let $r=C_T/C_W$ and
$d_n=(mC_W-nC_T)/\sqrt n$.

\begin{proposition}[Quantitative Gaussian comparison]
\label{prop:gaussian}
For every fixed $\Lambda<\infty$, uniformly over integer $m$ with
$|d_n|\le\Lambda$,
\begin{equation}
 \left|e_{\NS}(n,m)-\mathcal E(d_n)\right|
 =O\!\left(\frac{\log n}{\sqrt n}\right),
 \label{eq:quantitative}
\end{equation}
where $\mathcal E$ is given by \eqref{eq:profile}.
\end{proposition}

The uniform estimates needed below are proved in
Appendix~\ref{app:gaussian}. In the singleton-dispersion case, they give
the following bounds at testing budget $2^{-kC_K-\xi\sqrt k}$:
\begin{enumerate}
 \item Every permutation-invariant resource input has testing
 probability at most $1-\Phi(\xi/\sqrt{V_K})+O(\log k/\sqrt k)$.
 An iid capacity-achieving input attains this expression to the same
 precision.
 \item For every input type, the testing probability against its
 canonical product output reference is at most the same expression.
 \item For a rounded capacity-achieving type, the testing probability
 against \emph{every} output reference is at least the same expression.
\end{enumerate}
All remainders are uniform in the threshold, input types, and references
to which they apply. The second statement also applies to types with
a capacity deficit. It is an upper bound. It does not assert that all
types have the same Gaussian limit.

\subsection{Converse}
Take any NS protocol with error $\epsilon_n$. Round a capacity-achieving
input of $W$ to an $m$-type $\tau_m$ and restrict the target to this
class. Average the protocol over resource and target coordinate
permutations. The averaged channel is still NS. Memorylessness and
convexity of TV ensure that its error does not increase.
Under uniform target input $U_{\tau_m}$, the induced resource input
$P_{X^n}$ is permutation-invariant, but need not be a product law.

Apply Proposition~\ref{prop:converse} using
$R_{\mathcal Y,n}^{\mathrm u}$. At $h=nC_T+\xi\sqrt n\ge0$,
Lemmas~\ref{lem:uniform-source} and~\ref{lem:uniform-target} give
\begin{align*}
 \alpha_{2^{-h}}(\text{resource pair})
   &\le1-\Phi(\xi/\sqrt{V_T})+O(\log n/\sqrt n),\\
 \alpha_{2^{-h}}(\text{target pair})
   &\ge1-\Phi((\xi-d_n)/\sqrt{(m/n)V_W})
                           -O(\log n/\sqrt n).
\end{align*}
The second bound is valid for the protocol's induced target reference.
Since $m/n=r+O(n^{-1/2})$, the change from $(m/n)V_W$ to $rV_W$
costs $O(n^{-1/2})$ uniformly in $\xi$. This follows by differentiating
the normal CDF with respect to its positive variance.
Subtracting the two bounds and taking the supremum proves
\[
 \epsilon_n\ge\mathcal E(d_n)-O(\log n/\sqrt n).
\]
The restriction $\xi\ge-C_T\sqrt n$ excludes only an exponentially
small Gaussian tail. The bounds are uniform in the protocol. Taking
the infimum over protocols therefore proves the converse in
\eqref{eq:quantitative}.

\subsection{Achievability}
Reserve $\ell=\lceil A\log n\rceil$ of the resource uses for the
auxiliary messages, with $A$ large enough for
Proposition~\ref{prop:paid}, and put $k=n-\ell$.
The code error is $O(n^{-2})$.
For each target type, use an iid capacity-achieving resource input
in \eqref{eq:type-pairs}.
Write
\[
 d_k=\frac{mC_W-kC_T}{\sqrt k}
     =d_n+O(\log n/\sqrt n),\qquad
 \frac mk=r+O(n^{-1/2}).
\]
At $h=kC_T+\xi\sqrt k\ge0$, the first two uniform estimates give
\begin{align*}
 \alpha_{2^{-h}}(\text{resource pair})
   &\ge1-\Phi(\xi/\sqrt{V_T})-O(\log n/\sqrt n),\\
 \alpha_{2^{-h}}(\text{target pair})
   &\le1-\Phi((\xi-d_k)/\sqrt{(m/k)V_W})
                                      +O(\log n/\sqrt n).
\end{align*}
Normal CDFs are uniformly Lipschitz in their location and in their
variance on compact positive variance intervals. The last expression
may therefore be replaced by
$1-\Phi((\xi-d_n)/\sqrt{rV_W})$ at the stated precision.
Equation~\eqref{eq:testing-gap} bounds each type's deficiency by
$\mathcal E(d_n)+O(\log n/\sqrt n)$.
The testing budget zero causes no extra term: both signal laws are
absolutely continuous with respect to their canonical references.
The remainder is uniform in the target type, so
Proposition~\ref{prop:paid} proves achievability.
This proves the achievability bound in Proposition~\ref{prop:gaussian}.

\subsection{Inversion and resonance}
The function $\mathcal E$ is continuous and nondecreasing, with
limits zero and one at the two ends of the real line.
For $d$ in a fixed compact interval and $\mathcal E(d)$ bounded below
by a positive constant, every maximizing threshold in
\eqref{eq:profile} lies in one bounded interval.
On that interval the derivative
\[
 \frac1{\sqrt{rV_W}}\,
 \phi\!\left(\frac{\xi-d}{\sqrt{rV_W}}\right)
\]
has a positive lower bound for $d$ in a compact interval.
Evaluating the objective at a maximizing threshold for the smaller
displacement therefore gives
$\mathcal E(d_2)-\mathcal E(d_1)\ge c(d_2-d_1)$ locally at every
positive error level. Hence $\mathcal E(d_\epsilon)=\epsilon$
has a unique solution, with a locally Lipschitz inverse.

The inequality $\mathcal E(d)\le\epsilon$ is equivalent, after putting
$\eta=\Phi(\xi/\sqrt{V_T})$, to
\[
 d\le\sqrt{V_T}\Phi^{-1}(\eta)
          -\sqrt{rV_W}\Phi^{-1}(\eta-\epsilon)
 \quad\hbox{for every }\epsilon<\eta<1.
\]
Taking the infimum over $\eta$ gives \eqref{eq:quantile}.
Let $M_n=rn+(d_\epsilon/C_W)\sqrt n$.
For a sufficiently large fixed $H$, \eqref{eq:quantitative} and the
positive local slope imply
\[
 e_{\NS}(n,\lfloor M_n-H\log n\rfloor)<\epsilon
 <e_{\NS}(n,\lfloor M_n+H\log n\rfloor).
\]
Monotonicity in target length proves \eqref{eq:main-length}.
This argument bounds the remainder. It does not identify a third-order
logarithmic coefficient.

For comparison with the state-dichotomy literature, define
\[
 S_\nu(t)=\sup_x\{\Phi(x)-\Phi((x-t)/\sqrt\nu)\},
 \qquad \nu=\frac{rV_W}{V_T}.
\]
Then $\mathcal E(d)=S_\nu(d/\sqrt{V_T})$ and
$d_\epsilon=\sqrt{V_T}S_\nu^{-1}(\epsilon)$:
this is the sesquinormal form of the result
\cite{lipkabartosik2024dichotomies}.
When $\nu=1$, the maximizing threshold is $\xi=d/2$ for $d>0$,
and the supremum is zero for $d\le0$. This proves
\eqref{eq:resonance}. Moreover,
\[
 e_{\NS}(n,\lfloor rn\rfloor)
 =O(\log n/\sqrt n)
 \quad\text{under \eqref{eq:resonance-condition}}.
\]
This last estimate does not assert exact conversion at finite blocklength.

\appendix
\section{Nonunique capacity-achieving dispersions}
\label{app:nonunique}
We now allow different capacity-achieving dispersions, while retaining
positive capacities and positive minimum dispersions. Thus the equality
$V_{\min}=V_{\max}$ is no longer assumed.
Since \eqref{eq:dispersion} is affine
on the compact convex set of capacity-achieving inputs, its image is
the interval
\[
 \mathcal V_K=[V_{\min}(K),V_{\max}(K)].
\]
For a probability measure $\lambda$ on $\mathcal V_T$, set
$F_\lambda(\xi)=\int\Phi(\xi/\sqrt v)\lambda(dv)$ and define
\begin{equation}
 D(d,w)=\min_{\lambda\in\mathcal P(\mathcal V_T)}
 \sup_{\xi\in\mathbb R}
 \left\{F_\lambda(\xi)-\Phi\!\left(\frac{\xi-d}{\sqrt{rw}}\right)\right\},
 \qquad
 \mathcal E_{\rm gen}(d)=\max_{w\in\mathcal V_W}D(d,w).
 \label{eq:general-profile}
\end{equation}
Here probability measures on the compact interval are Borel measures.

\begin{theorem}[General positive-dispersion formula]
\label{thm:general}
If $C_T,C_W,V_{\min}(T),V_{\min}(W)>0$, then, uniformly for bounded
$d_n=(mC_W-nC_T)/\sqrt n$,
\[
 e_{\NS}(n,m)=\mathcal E_{\rm gen}(d_n)+O(\log n/\sqrt n).
\]
For every $0<\epsilon<1$ there is a unique
$d_\epsilon$ with $\mathcal E_{\rm gen}(d_\epsilon)=\epsilon$, and
\[
 m^*_{\NS}(n,\epsilon)
 =rn+\frac{d_\epsilon}{C_W}\sqrt n+O(\log n).
\]
\end{theorem}
\begin{proof}
We follow the proof in Section~\ref{sec:second-order}, now allowing the
variance choices to depend on the target type.
For achievability, associate each target type $\tau$ with a closest
capacity-achieving input $Q_\tau$ and its variance $w_\tau$.
Choose a minimizing $\lambda_\tau$ in $D(d_n,w_\tau)$.
Lemma~\ref{lem:uniform-source} realizes it by a mixture of iid
capacity-achieving inputs. For a given target type, the same mixture
works at every testing threshold. Lemma~\ref{lem:uniform-type} bounds the target
testing curve using $w_\tau$. Thus its deficiency is at most
$D(d_n,w_\tau)+O(\log n/\sqrt n)$.
The type header permits these choices to depend on $\tau$. The largest
type error is therefore at most $\mathcal E_{\rm gen}(d_n)$ plus the
same remainder.

For the converse, fix $w\in\mathcal V_W$, round a capacity-achieving
input with variance $w$ to a target type, and average the protocol as before.
Lemma~\ref{lem:uniform-source} assigns its correlated resource input
one measure $\lambda_n$ independent of the testing threshold.
Lemma~\ref{lem:uniform-target} bounds the target curve for every
induced reference. Consequently its error is at least
$D(d_n,w)-O(\log n/\sqrt n)$. Taking the maximum over $w$ gives the
converse bound.
All errors are uniform over the compact variance intervals.
No exchange of minimum and maximum is used.

It remains to justify inversion of the limiting error profile.
For bounded $d$, the difference in braces in \eqref{eq:general-profile}
tends uniformly to zero as $|\xi|\to\infty$, over all $w,\lambda$.
On compact $\xi$ intervals it is continuous in $d,w$ and in $\lambda$
for weak convergence. Compactness of probability measures on
$\mathcal V_T$ therefore gives attainment of the minimum and
continuity of $D$ and $\mathcal E_{\rm gen}$.
As $d\to-\infty$, splitting at $\xi=d/2$ bounds the supremum by
vanishing Gaussian tails. As $d\to+\infty$, evaluation at $\xi=d/2$
makes the difference tend uniformly to one.

To establish the local inverse bound where
$\mathcal E_{\rm gen}\ge\epsilon_0>0$, take a maximizing $w_*$ at
$d_1$. For every $\lambda$, its threshold supremum is at least
$\epsilon_0$ and is attained in one common compact interval.
Increasing $d_1$ to $d_2$ raises the objective at that maximizing
threshold by at least $c(d_2-d_1)$, with $c>0$ uniform in
$w_*,\lambda$. Minimize over $\lambda$ and use the same $w_*$ at $d_2$.
This proves the local inverse bound. The integer bracketing argument
of Section~\ref{sec:second-order} then gives the length expansion.
\end{proof}
The optimization order $\max_w\min_\lambda\sup_\xi$ has an operational
meaning: the target type is known to the encoder and communicated
to the decoder, whereas the resource input must be chosen before the
testing threshold. The optimization order retains this dependence.

\section{Uniform Gaussian estimates}
\label{app:gaussian}

This appendix proves the uniform estimates used in the converse and
achievability arguments. Let $K:\mathcal A\to\mathcal B$ be a finite channel
with $C_K>0$ and $V_{\min}(K)>0$. Constants below may depend on $K$,
but not on input types, testing thresholds, or reference distributions.
Products such as $P_{A^k}K^{\otimes k}$ in this appendix denote joint input--output laws, not output marginals.
For a $k$-type $P$ and any $a^k\in\mathcal T_P^k$, define
\begin{equation}
 J_{P,k}:=\sum_{i=1}^k\log\frac{K(B_i|a_i)}{(PK)(B_i)},
 \qquad B^k\sim K^{\otimes k}(\cdot|a^k).
 \label{eq:appendix-type-density}
\end{equation}
Its distribution depends only on $P$, not on the chosen word. Its mean
and variance are $kI(P,K)$ and $kV(P,K)$, respectively.

We first establish the moment and regularity bounds needed below.
The second and third absolute moments of $i_{P,K}(A;B)$ are bounded
uniformly over all $P$.
Indeed, if $K_{\min}$ is the smallest positive channel entry, then on
events of positive probability
\[
 \log K_{\min}\leq i_{P,K}(a;b)\leq\log(1/P(a)),
\]
and $\sum_aP(a)\log^j(1/P(a))$ is bounded for $j=2,3$.
Thus $V(P,K)$ is bounded on the simplex. The Berry--Esseen theorem applied to the
independent summands in \eqref{eq:appendix-type-density} gives
\begin{equation}
 \sup_u\left|\Pr\{J_{P,k}\leq kI(P,K)+u\sqrt k\}
       -\Phi\!\left(\frac u{\sqrt{V(P,K)}}\right)\right|
 \leq\frac{c_K}{\sqrt k\,V(P,K)^{3/2}}
 \label{eq:appendix-berry-esseen}
\end{equation}
whenever $V(P,K)>0$. See also
\cite[Theorem~44 and Lemma~62]{polyanskiy2010channel}.
Taking left limits gives the corresponding upper-tail bounds, with
equality at the threshold included.

All capacity-achieving inputs have the same output law, which is
positive on every output reachable through $K$. After deleting
unreachable outputs, this law is strictly positive. Consequently,
$I(P,K)$ and $V(P,K)$ have bounded derivatives on a fixed neighbourhood
of the capacity-achieving set $\Pi_K$. The variance here is the
conditional variance, whose finite-sum formula remains smooth even
when some input probabilities vanish. Finally, the quadratic capacity-loss
bound of \cite{caotomamichel2023quadratic} gives constants $c,L,\delta>0$
such that, with $t=\operatorname{dist}_2(P,\Pi_K)$ and a closest
$P^*\in\Pi_K$,
\begin{equation}
 C_K-I(P,K)\geq ct^2,\qquad
 |V(P,K)-V(P^*,K)|\leq Lt,\qquad
 V(P,K)\geq V_{\min}(K)/2
 \label{eq:appendix-capacity-neighbourhood}
\end{equation}
for $t\leq\delta$. These estimates provide uniform constants below.
The resource converse will not require a product-input assumption.

\begin{lemma}[Uniform conditional-type tail]
\label{lem:uniform-type}
For each input distribution $P$, choose a closest $P^*(P)\in\Pi_K$ and
put $v(P)=V(P^*(P),K)$. There is a constant $c$ such that, for every
$k\geq2$, every $k$-type $P$, and every real $\xi$,
\begin{equation}
 \Pr\{J_{P,k}\geq kC_K+\xi\sqrt k\}
 \leq1-\Phi\!\left(\frac\xi{\sqrt{v(P)}}\right)
       +c\frac{\log k}{\sqrt k}.
 \label{eq:uniform-type-tail}
\end{equation}
\end{lemma}
\begin{proof}
Choose a constant $M$ bounding all conditional variances and all
capacity-achieving variances, and put $H_k=\sqrt{4M\ln k}$.
For $t\leq\delta$, the Berry--Esseen remainder in
\eqref{eq:appendix-berry-esseen} is $O(k^{-1/2})$ uniformly in $P$.
Lipschitz continuity of the reciprocal square-root on the positive
variance interval and \eqref{eq:appendix-capacity-neighbourhood} give
\[
 \frac{\xi+\sqrt k(C_K-I(P,K))}{\sqrt{V(P,K)}}
       -\frac\xi{\sqrt{v(P)}}
 \geq c_1\sqrt k\,t^2-c_2|\xi|t
 \geq-c_3\frac{\xi^2}{\sqrt k}.
\]
For $|\xi|\leq H_k$, the Lipschitz bound for $\Phi$ therefore proves
\eqref{eq:uniform-type-tail}. For $\xi>H_k$, the actual mean is at most
$kC_K$, so the probability on the left is at most
$1-\Phi(H_k/\sqrt M)+O(k^{-1/2})=O(k^{-1/2})$.
For $\xi<-H_k$, the Gaussian tail on the right is at least $1-k^{-2}$,
and the trivial upper bound one suffices.

Outside this neighbourhood, compactness and continuity give a fixed
capacity gap $C_K-I(P,K)\geq g>0$. For $\xi\geq-H_k$, the threshold
exceeds the mean by at least $kg/2$ once $k$ is large. The variance
bound and Chebyshev's inequality give a probability $O(k^{-1})$.
For $\xi<-H_k$, the trivial bound again suffices. Enlarging $c$ covers
the finitely many smaller blocklengths.
\end{proof}

\begin{lemma}[Resource estimates and realization of Gaussian mixtures]
\label{lem:uniform-source}
Let $k\geq2$, let $P_{A^k}$ be any permutation-invariant input law, and set
\[
 p_k=P_{A^k}K^{\otimes k},\qquad
 q_k=P_{A^k}\otimes R^{\mathrm u}_{\mathcal B,k}.
\]
There is a probability measure $\lambda_k$ on $\mathcal V_K$ such that,
for all $\xi$ with $h=kC_K+\xi\sqrt k\geq0$,
\begin{equation}
 \alpha_{2^{-h}}(p_k,q_k)
 \leq1-\int_{\mathcal V_K}\Phi(\xi/\sqrt v)\,\lambda_k(dv)
       +c\frac{\log k}{\sqrt k}.
 \label{eq:uniform-source-upper}
\end{equation}
One measure works for all thresholds simultaneously. Conversely, for every
probability measure $\lambda$ on $\mathcal V_K$, there is a
permutation-invariant input $P_{A^k}$ for which
\begin{equation}
 \left|\alpha_{2^{-h}}(p_k,q_k)
       -1+\int_{\mathcal V_K}\Phi(\xi/\sqrt v)\,\lambda(dv)\right|
 \leq c\frac{\log k}{\sqrt k}.
 \label{eq:uniform-source-realization}
\end{equation}
The constant $c$ is independent of the input law, $\lambda$, and $\xi$.
\end{lemma}
\begin{proof}
Decompose the permutation-invariant input law into its type classes:
\[
 P_{A^k}=\sum_{P\in\mathcal P_k(\mathcal A)}
       \theta(P)U_{\mathcal T_P^k},
 \qquad \theta(P)=P_{A^k}(\mathcal T_P^k).
\]
Write $N_k=|\mathcal P_k(\mathcal B)|$. Every product output law is
bounded pointwise by $N_k R^{\mathrm u}_{\mathcal B,k}$. Conditional
on input type $P$, the log-likelihood-ratio for $(p_k,q_k)$ therefore satisfies
\[
 \widetilde J:=\log\frac{K^{\otimes k}(B^k|A^k)}
                                  {R^{\mathrm u}_{\mathcal B,k}(B^k)}
 \leq J_{P,k}+\log N_k.
\]
Use the upper spectrum bound \eqref{eq:spectrum} with $\gamma=\log k$,
then condition its likelihood-ratio event on the input type. This gives
\[
 \alpha_{2^{-h}}(p_k,q_k)
 \leq\sum_P\theta(P)\Pr\{J_{P,k}\geq h-\log k-\log N_k\}+k^{-1}.
\]
Lemma~\ref{lem:uniform-type} and the common Lipschitz constant of the
Gaussian CDFs, whose variances are at least $V_{\min}(K)$, now prove
\eqref{eq:uniform-source-upper} with
$\lambda_k=\sum_P\theta(P)\delta_{v(P)}$. We have used
$\log N_k=O(\log k)$. Here we condition only the likelihood-ratio event
in the spectrum bound. We do not optimize separate tests for different
input types.

To realize a prescribed mixture, choose capacity-achieving inputs
$P_-,P_+$ attaining the two endpoint variances. Their output laws agree, and the conditional
variance is affine on $\Pi_K$. Linear interpolation therefore supplies
a continuous choice $P_v\in\Pi_K$ with $V(P_v,K)=v$ for every
$v\in\mathcal V_K$. When the interval is a singleton, use a constant
choice. Set
\[
 P_{A^k}=\int_{\mathcal V_K}P_v^{\otimes k}\,\lambda(dv).
\]
Its output law is $(P_B^*)^{\otimes k}$, where $P_B^*$ is the common
capacity-achieving output law. Conditional on the sampled $v$, the
random variable
\[
 J=\log\frac{K^{\otimes k}(B^k|A^k)}{(P_B^*)^{\otimes k}(B^k)}
\]
is a sum of iid information densities with mean $C_K$ and variance $v$.
Here unconditional and conditional variances coincide: each input
symbol used by a capacity-achieving law has row divergence $C_K$.
The uniform third-moment bound and $v\geq V_{\min}(K)>0$ give
\begin{equation}
 \sup_u\left|\Pr\{J\leq kC_K+u\sqrt k\}
       -\int\Phi(u/\sqrt v)\,\lambda(dv)\right|
 \leq c_0k^{-1/2}.
 \label{eq:appendix-mixture-normal}
\end{equation}
Apply the Berry--Esseen theorem conditionally on $v$ and then integrate. Its
constant is uniform in the measure $\lambda$.

Product domination gives $\widetilde J\leq J+\log N_k$. Since the actual
output law is $(P_B^*)^{\otimes k}$, Markov's inequality also gives
\[
 \Pr\!\left\{
 \frac{R^{\mathrm u}_{\mathcal B,k}(B^k)}{(P_B^*)^{\otimes k}(B^k)}>k
 \right\}\leq k^{-1}.
\]
Thus $\widetilde J\geq J-\log k$ outside an event of probability $k^{-1}$.
The two spectrum bounds imply
\[
 \Pr\{J\geq h+\log k\}-k^{-1}
 \leq\alpha_{2^{-h}}(p_k,q_k)
 \leq\Pr\{J\geq h-\log k-\log N_k\}+k^{-1}.
\]
Apply \eqref{eq:appendix-mixture-normal} to these tails. The Gaussian
mixture distribution functions have Lipschitz constant at most
$(2\pi V_{\min}(K))^{-1/2}$, so the threshold shifts change their values
by $O(\log k/\sqrt k)$ uniformly in $\lambda$ and $\xi$. This proves
\eqref{eq:uniform-source-realization}.
\end{proof}

\begin{lemma}[Target estimate for every output reference]
\label{lem:uniform-target}
Fix $\eta<\infty$ and $k\geq2$. Let $P^*\in\Pi_K$ and let $P_k$ be a $k$-type
with $\|P_k-P^*\|_2\leq\eta/k$. Put
\[
 p_k=U_{\mathcal T_{P_k}^k}K^{\otimes k},\qquad
 q_k=U_{\mathcal T_{P_k}^k}\otimes R_k,
\]
where $R_k$ is any output distribution. With $v=V(P^*,K)$, uniformly
for $h=kC_K+\xi\sqrt k\geq0$,
\begin{equation}
 \alpha_{2^{-h}}(p_k,q_k)
 \geq1-\Phi(\xi/\sqrt v)-c_{K,\eta}\frac{\log k}{\sqrt k}.
 \label{eq:uniform-target-lower}
\end{equation}
\end{lemma}
\begin{proof}
For any $k$-type $P$, the product law with marginal $P$ assigns its type
class mass at least $(k+1)^{-|\mathcal A|}$. Conditioning this law on its type class
gives the uniform law. The corresponding channel output marginal
therefore satisfies
\[
 (U_{\mathcal T_P^k}K^{\otimes k})(b^k)
 \leq(k+1)^{|\mathcal A|}(PK)^{\otimes k}(b^k).
\]
Write $a_k=|\mathcal A|\log(k+1)$. For every $R_k$ and $\gamma>0$,
summing over the exceptional outputs yields
\begin{equation}
 \Pr\!\left\{
 \log\frac{R_k(B^k)}{(PK)^{\otimes k}(B^k)}>a_k+\gamma
 \right\}\leq2^{-\gamma}.
 \label{eq:appendix-arbitrary-reference}
\end{equation}
Outputs outside the support of $(PK)^{\otimes k}$ have zero actual
probability. Where $R_k(B^k)=0$, the log-likelihood-ratio for $(p_k,q_k)$
is $+\infty$, and the following lower bound still holds. The lower
spectrum bound and \eqref{eq:appendix-arbitrary-reference} give
\[
 \alpha_{2^{-h}}(p_k,q_k)
 \geq\Pr\{J_{P_k,k}\geq h+a_k+\log k\}-k^{-1}.
\]
The derivative bounds near $\Pi_K$ imply, uniformly in $P^*$ and $P_k$,
\[
 k(C_K-I(P_k,K))=O(1),\qquad
 V(P_k,K)=v+O(k^{-1}).
\]
In particular, the actual variance stays bounded away from zero for
large $k$. Equation~\eqref{eq:appendix-berry-esseen} therefore compares
the last tail with
$1-\Phi((\xi+\Delta_k)/\sqrt{V(P_k,K)})$, up to $O(k^{-1/2})$,
where $0\leq\Delta_k=O(\log k/\sqrt k)$.
This comparison is uniform even for unbounded $\xi$: for variances
$s$ in a compact subset of $(0,\infty)$, both
\[
 \left|\frac{\partial}{\partial u}\Phi(u/\sqrt s)\right|
 =\frac{\phi(u/\sqrt s)}{\sqrt s},\qquad
 \left|\frac{\partial}{\partial s}\Phi(u/\sqrt s)\right|
 =\frac{|u|\phi(u/\sqrt s)}{2s^{3/2}}
\]
are uniformly bounded over $u\in\mathbb R$. Replacing the location
and variance thus costs $O(\log k/\sqrt k)$, proving the claim for
large $k$. Increasing the constant covers the finitely many smaller
$k\geq2$. Every bound is independent of $R_k$, as required for the
protocol-dependent reference in the converse.
\end{proof}

\section{Shared randomness: agreement and separation}\label{app:sr}
We give examples with positive capacities and dispersions for which
the shared-randomness second-order coefficient agrees with the
non-signalling value, and prove a strict second-order separation for
another pair.
The agreement examples also show that resonance is possible with
shared randomness. Matching the total information variances, however,
does not suffice in general. Whenever the SR and NS values agree,
the EA value agrees as well, since
$\mathrm{SR}\subseteq\mathrm{EA}\subseteq\mathrm{NS}$.

\subsection{Shared-randomness resonance for packet-erasure channels}
\label{app-sr:packet-resonance}
Let $E_a(\delta)$ be the erasure channel with input alphabet
$\{0,1\}^a$: it transmits the entire $a$-bit input with probability
$1-\delta$ and otherwise outputs one erasure symbol. Thus
\[
 C_{E_a(\delta)}=a(1-\delta),\qquad
 V_{E_a(\delta)}=a^2\delta(1-\delta).
\]
For $a>1$, the channel erases all $a$ bits together, unlike a product
of $a$ independent binary erasure channels.

\begin{proposition}\label{app-sr:packet}\label{app-sr:bec}
Fix an integer $a\ge1$ and $0<\delta_T<\delta_W<1$. Let
$T=E_a(\delta_T)$ and $W=\mathrm{BEC}(\delta_W)$, and set
\[
 \alpha=1-\delta_T,\quad \beta=1-\delta_W,\quad
 r=\frac{a\alpha}{\beta},\quad
 \sigma_T=a\sqrt{\alpha(1-\alpha)},\quad
 \sigma_W=\sqrt{a\alpha(1-\beta)}.
\]
For $A\in\{\SR,\EA,\NS\}$ and every fixed $0<\epsilon<1$,
\begin{align}
 m_A^*(n,\epsilon;T,W)
 &=rn+\frac{d_\epsilon}{\beta}\sqrt n+o(\sqrt n),
 \label{app-sr:packet-rate}\\
 d_\epsilon&=\inf_{\epsilon<u<1}
 \{\sigma_T\Phi^{-1}(u)-\sigma_W\Phi^{-1}(u-\epsilon)\}.
 \nonumber
\end{align}
For every $m\ge an$, there is an SR protocol with error exactly
\begin{equation}\label{app-sr:packet-finite}
 \Delta_{n,m}=\max_{t\in\mathbb Z}
 \{\Pr[aK\le t]-\Pr[L\le t]\},
 \qquad K\sim\operatorname{Bin}(n,\alpha),\quad
 L\sim\operatorname{Bin}(m,\beta).
\end{equation}
The displayed binomial laws specify the two marginal distributions.
The protocol couples these counts dependently.
\end{proposition}
\begin{proof}
Using shared randomness, choose a uniform ordered sample of $an$
distinct target positions. Pack their input bits, $a$ at a time, into
the $n$ resource symbols.
Conditional on $K=k$ surviving resource packets, the receiver knows
the correct bits on a uniform $ak$-subset of the $m$ target positions.
This statement holds for every target input word, since the random
injection and the erasure pattern are independent of its values.

Write $Q_{aK},Q_L$ for generalized quantiles, and put
$\Delta=\Delta_{n,m}$. The coupling
\[
 aK=Q_{aK}(U),\qquad L'=Q_L((U-\Delta)\bmod1),
 \qquad U\text{ uniform on }[0,1),
\]
satisfies $\Pr\{L'>aK\}=\Delta$. Indeed, for $U\ge\Delta$,
$F_L(Q_{aK}(U))\ge U-\Delta$, hence $L'\le aK$. Every coupling
has failure probability at least $F_{aK}(t)-F_L(t)$ for each $t$.
The receiver implements this coupling by sampling $U$ uniformly within
the quantile interval corresponding to its observed count. It then
reveals a uniformly chosen subset of size $J=\min(aK,L')$ of the
received bits. This requires no communication back to the encoder.
Conditional on $J$, the revealed set is uniform and all its values
are correct, exactly as for the target channel conditional on its
visible count. Thus the worst-input channel error equals
$\TV(\mathcal L(J),\mathcal L(L))=\Delta$: the upper bound follows
from the coupling, and the reverse bound from $J\le aK$.

Because $r>a$, the condition $m\ge an$ holds throughout any fixed
second-order window for all sufficiently large $n$. At
$m=\lfloor rn+d\sqrt n/\beta\rfloor$, uniform Berry--Esseen estimates
give
\begin{equation}\label{app-sr:packet-profile}
 \Delta_{n,m}\longrightarrow
 D(d):=\sup_x\{\Phi(x/\sigma_T)
                         -\Phi((x-d)/\sigma_W)\}.
\end{equation}
Theorem~\ref{thm:main} gives the same NS converse, since
$\sigma_T^2=V_T$ and $\sigma_W^2=rV_W$. The inclusions
$\SR\subseteq\EA\subseteq\NS$ therefore give the same limiting
error for all three assistance classes. The positive branch of $D$
is strictly increasing, and inversion gives
\eqref{app-sr:packet-rate}.
\end{proof}

Thus shared randomness attains the threshold error at the capacity
ratio. To evaluate this threshold, put
\begin{equation}\label{app-sr:packet-threshold}
 \nu=\frac{\sigma_W^2}{\sigma_T^2}
     =\frac{\delta_W}{a\delta_T},\qquad
 \epsilon_*=D(0)=
 \left|\Phi(x_\nu)-\Phi(x_\nu/\sqrt\nu)\right|,
 \qquad
 x_\nu=\sqrt{\frac{\nu\ln\nu}{\nu-1}}.
\end{equation}
At $\nu=1$ set $\epsilon_*=0$. If $\nu\ne1$, then
$d_{\epsilon_*}=0$, and the optimal rate at this fixed positive
threshold satisfies
\[
 \frac{m_A^*(n,\epsilon_*)}{n}=r+o(n^{-1/2}),
 \qquad A\in\{\SR,\EA,\NS\}.
\]
For $a>1$, the error at $m=\lfloor rn\rfloor$ tends to zero when
$\delta_T=\delta_W/a$. This gives zero-threshold resonance for two
fixed, distinct channels at an interior point of
$0<\delta_T<\delta_W$. The statement concerns vanishing asymptotic
error. It does not require exact conversion at finite blocklength.

For example, fix $W=\mathrm{BEC}(1/2)$ and let
$T=E_2(\delta)$, $0<\delta<1/2$. Then
\[
 r=4(1-\delta),\qquad \nu=\frac1{4\delta},
\]
and \eqref{app-sr:packet-threshold} has its zero at $\delta=1/4$.
The interval where $\epsilon_*\le0.01$ has endpoints approximately
$0.230$ and $0.272$, the solutions of
\eqref{app-sr:packet-threshold} with $\epsilon_*=0.01$.
The resulting threshold-error curve is analogous to Fig.~2 of
\cite{lipkabartosik2024dichotomies}, with the same value for SR, EA,
and NS assistance.

At the resonance point, the resource has capacity $3/2$ and variance
$3/4$, whereas the target has capacity $1/2$ and variance $1/4$.
Thus $r=3$ and $V_T=rV_W=3/4$. In particular,
\[
 m_A^*(n,\epsilon)
 =3n+2\sqrt3\,
   \Phi^{-1}\!\left(\frac{1+\epsilon}{2}\right)\sqrt n+o(\sqrt n)
\]
for every fixed positive $\epsilon$. The finite-blocklength protocol
at $m=3n$ uses only random injection and thinning. Its error has the
exact expression
\begin{equation}\label{app-sr:packet-exact-counts}
 \Delta_{n,3n}=2^{-3n}\max_{0\le t\le3n}
 \left\{2^n\sum_{k=0}^{\min(n,\lfloor t/2\rfloor)}
                    \binom nk3^k
              -\sum_{j=0}^{t}\binom{3n}j\right\}.
\end{equation}
All terms in braces are integers. Exact evaluation gives
\[
 9.718\times10^{-3}<\Delta_{250,750}<9.719\times10^{-3}<0.01.
\]
Thus $250$ resource uses simulate $750$ target uses at the capacity
ratio with error below one percent. At a shorter blocklength, a target
length one below the capacity-ratio length also suffices:
\[
 \begin{gathered}
 \Delta_{16,47}=\frac{701606824869}{70368744177664},\\
 9.970\times10^{-3}<\Delta_{16,47}<9.971\times10^{-3}<0.01.
 \end{gathered}
\]
The target length is $47=3\cdot16-1$, giving rate $47/16\simeq2.938$,
only $1/48$ below the capacity ratio in relative terms.
The script \texttt{verify\_packet\_resonance.py} verifies both maxima
and decimal brackets using rational arithmetic. These are the exact
errors of the specified protocol. They give achievable SR bounds.
Optimality at finite blocklength is not claimed.

The case $a=1$ recovers conversion between distinct binary erasure
channels in the direction $\delta_T<\delta_W$. The reverse direction
remains outside this packing argument, since the capacity-ratio
target length is then smaller than $an$ and the required injection
is unavailable.

\subsection{Binary-symmetric self-conversion}
\begin{proposition}\label{app-sr:self}
For $T=W=\mathrm{BSC}(q)$, $0<q<1/2$, and fixed $b\in\mathbb R$,
\begin{equation}\label{app-sr:self-error}
 \lim_{n\to\infty}e_A(n,\lfloor n+b\sqrt n\rfloor)
 =\left[2\Phi\!\left(\frac{bC_q}{2\sqrt{V_q}}\right)-1\right]_+,
 \qquad A\in\{\SR,\EA,\NS\}.
\end{equation}
Consequently all three assistance classes have the expansion
\begin{equation}\label{app-sr:self-rate}
 m_A^*(n,\epsilon)=n+
 \frac{2\sqrt{V_q}}{C_q}
 \Phi^{-1}\!\left(\frac{1+\epsilon}{2}\right)\sqrt n+o(\sqrt n).
\end{equation}
\end{proposition}
\begin{proof}
For $b\le0$, discard resource uses. Suppose $b>0$ and fix
$q<t<1/2$ and $1<\rho<C_q/C_t$. Channel coding followed by
shared-randomness simulation converts $N$ uses of $\mathrm{BSC}(q)$
to $\lfloor\rho N\rfloor$ uses of $\mathrm{BSC}(t)$ with vanishing
worst-input error as $N\to\infty$
\cite{polyanskiy2010channel,cao2024channel}. Choose both component
rates strictly within their first-order limits. Apply this construction
to
\[
 N_n=\frac{b\sqrt n}{\rho-1}+O(1),\qquad L_n=\rho N_n+O(1)
\]
resource and target positions, and use the other $n-N_n$ resource
coordinates directly. Choose the integer roundings to obtain at least
$\lfloor n+b\sqrt n\rfloor$ target positions, with only $O(1)$ excess,
and discard the excess after a shared uniform permutation.

Up to $o(1)$ TV error, the simulated noise before discarding is a
uniform permutation of independent Bernoulli noise: parameter $q$ on
$n-N_n$ positions and parameter $t$ on $L_n$ positions. Both it and the
ideal target noise are uniform conditional on Hamming weight. Their
TV distance is therefore the distance between the weight laws.
For independent $A_M\sim\operatorname{Bin}(M,q)$,
$B_M\sim\operatorname{Bin}(L_M,t)$ and
$D_M\sim\operatorname{Bin}(L_M,q)$, with $L_M/\sqrt M\to a$,
\begin{equation}\label{app-sr:binomial-shift}
 \TV(A_M+B_M,A_M+D_M)
 \longrightarrow
 2\Phi\!\left(\frac{a|t-q|}{2\sqrt{q(1-q)}}\right)-1.
\end{equation}
We justify this total variation limit directly. Binomial unimodality
gives
\[
 \TV(A_M,A_M+1)=\max_j\Pr\{A_M=j\}=O(M^{-1/2}).
\]
Replacing $B_M,D_M$ by nearest integers to their means costs
$O(\sqrt{(L_M+1)/M})=o(1)$ in TV. For an integer shift
$h_M/\sqrt M\to a(t-q)$, write $p_M$ for the mass function of
$A_M$. Stirling's formula on central windows gives
\[
 \ln\frac{p_M(A_M-h_M)}{p_M(A_M)}
 \ \Longrightarrow\ N(-v^2/2,v^2),\qquad
 v=\frac{a(t-q)}{\sqrt{q(1-q)}},
\]
and boundary events vanish. Applying the bounded continuous function
$x\mapsto(1-e^x)_+$ proves \eqref{app-sr:binomial-shift}.
Thus the weight comparison holds in total variation, as required for
the channel error, rather than only in distribution.

Our protocol thus has limiting error at most
\[
 2\Phi\!\left(
 \frac{b\rho(t-q)}{2(\rho-1)\sqrt{q(1-q)}}\right)-1.
\]
First let $\rho\uparrow C_q/C_t$, then $t\downarrow q$. Since
$C_q-C_t=\ell_q(t-q)+O((t-q)^2)$, this becomes the right side of
\eqref{app-sr:self-error}. The parameters $t$ and $\rho$ remain fixed
while $n\to\infty$, so the argument requires only coding and
simulation at fixed rates below their respective first-order limits.
The NS converse gives equality, and monotonicity in target length
gives \eqref{app-sr:self-rate}.
\end{proof}

The construction preserves the resource noise in all but $O(\sqrt n)$
coordinates. Coding and simulation are needed only on this smaller
block. Direct use supplies the remaining target coordinates.

\subsection{A tighter shared-randomness converse: BEC-to-BSC}
\label{app:sr-converse}
For this final example, let $T=\mathrm{BEC}(1-\alpha)$ and
$W=\mathrm{BSC}(q)$, with $0<\alpha<1$ and $0<q<1/2$.
Thus $\alpha$ is the probability that a resource bit is visible.
Put $r=\alpha/C_q$ and
$m_n(d)=\lfloor rn+d\sqrt n/C_q\rfloor$. An SR protocol uses the
same encoder for every erasure pattern. We first derive a constraint
that follows from this requirement.
The auxiliary-source construction follows the binary-Hamming
erasure-broadcast converse of \cite{kochman2020distortion}.
We account explicitly for conditioning on simultaneous success and
then average over ranges of erasure counts in the Gaussian limit.
The resulting converse is strictly stronger than the NS bound,
although it does not determine the optimal SR coefficient.

\begin{lemma}[Conditioned broadcast bound]\label{app-sr:broadcast}
Let $S$ be uniform on $\{0,1\}^m$, and let $R$ contain all independent
protocol seeds, so that $X^n$ is a deterministic function of $(S,R)$.
Choose uniform nested masks $A_1\subset A_2$, of sizes $k_1<k_2$,
independently of $(S,R)$. Set $B_i=X_{A_i}$ and
$C_0=(A_1,A_2,R)$. Suppose the two reconstructions
$\widehat S_i(B_i,C_0)$ obey
\[
 G=\{d_H(S,\widehat S_i)\le md_i\text{ for }i=1,2\},
 \qquad\Pr(G)=1-\eta>0,
\]
where $0\le d_i\le1/2$. For $\lambda=k_1/k_2$,
$\Delta=k_2-k_1$ and $0<\theta<1/2$,
\begin{align}\label{app-sr:finite}
 &mC_{\theta\star d_1}
 +\lambda m\{h(\theta\star d_2)-h(d_2)\}-k_1\nonumber\\
 &\quad\le
 \frac{\eta\lambda\Delta+(1-\lambda)h(\eta)}{1-\eta}
 +(1+\lambda)\log\frac1{1-\eta},
 \qquad \theta\star t=\theta+(1-2\theta)t.
\end{align}
\end{lemma}
\begin{proof}
Write $P$ for the original law, $Q=P(\cdot\mid G)$ and
$D_0=\log(1/(1-\eta))$. After conditioning, generate
$U=S\oplus N$ with fresh independent
$N\sim\operatorname{Ber}(\theta)^{\otimes m}$.
Since $S,U$ are uniform and independent of $C_0$ under $P$,
$D(Q\Vert P)=D_0$ implies
$H_Q(S\mid C_0),H_Q(U\mid C_0)\ge m-D_0$.
The binary Fano inequality and conditional vector Mrs.\ Gerber's
lemma \cite{wyner1973entropy} give, respectively,
\begin{align*}
 I_Q(U;B_1\mid C_0)&\ge mC_{\theta\star d_1}-D_0,\\
 I_Q(S;B_2\mid U,C_0)&\ge
 m\{h(\theta\star d_2)-h(d_2)\}-D_0.
\end{align*}
For the second inequality, use
$H_Q(S\mid B_2,C_0)\le mh(d_2)$ and the fact that
$h(\theta\star h^{-1}(t))-t$ is decreasing on $[0,1]$.
$h^{-1}$ is taken in $[0,1/2]$.

For any law $L$, define
$\mathcal D_L=H_L(B_1\mid U,C_0)-\lambda H_L(B_2\mid U,C_0)$.
Under $P$, uniform masks independent of $(X^n,U,R)$ give
$\mathcal D_P\ge0$: average subset-entropy per coordinate decreases
with subset size, by Han's subset-entropy inequality
\cite{han1978nonnegative}, or directly by the entropy chain rule
along a uniform random ordering.
For $Q^-=P(\cdot\mid G^c)$ the exact mixture identity is
\begin{align*}
 \mathcal D_P={}&(1-\eta)\mathcal D_Q+\eta\mathcal D_{Q^-}\\
 &+I_P(B_1;\mathbf1_G\mid U,C_0)
 -\lambda I_P(B_2;\mathbf1_G\mid U,C_0).
\end{align*}
As $B_1$ is a function of $(B_2,C_0)$, the second line is at most
$(1-\lambda)h(\eta)$. Also
$\mathcal D_{Q^-}\le(1-\lambda)k_1=\lambda\Delta$.
Thus
\[
 \mathcal D_Q\ge-
 \frac{\eta\lambda\Delta+(1-\lambda)h(\eta)}{1-\eta}.
\]
Finally, determinism of $B_2$ given $(S,C_0)$ implies
\[
 I_Q(U;B_1\mid C_0)+\lambda I_Q(S;B_2\mid U,C_0)
 =H_Q(B_1\mid C_0)-\mathcal D_Q\le k_1-\mathcal D_Q.
\]
Combining these bounds proves the claim. If $\eta=0$, use
$\mathcal D_P\ge0$ directly. The independence of the masks is used
only under the original law $P$. The mixture identity accounts for
its loss under conditioning.
\end{proof}

Define
\begin{equation}\label{app-sr:parameters}
 f_\theta=\frac{C_{\theta\star q}}{C_q},\qquad
 a_\theta=\frac{(1-2\theta)\ell_{\theta\star q}}{\ell_q},\qquad
 \kappa_\theta=f_\theta-a_\theta>0.
\end{equation}
To prove the strict inequality, use the positive series
\[
 A(t)=C_{(1-t)/2}
 =\frac1{\ln2}\sum_{j\ge1}\frac{t^{2j}}{2j(2j-1)}.
\]
For $z\in(0,1)$, $A(zt)/A(t)$ is a weighted average of $z^{2j}$.
Its derivative with respect to $\ln t$ is
$2\operatorname{Cov}(j,z^{2j})<0$. Hence $f_\theta$ increases with
$q$, and direct differentiation gives
$f_\theta'(q)=\ell_q(f_\theta-a_\theta)/C_q>0$.
In fact $0<a_\theta<f_\theta<1$.

\begin{proposition}[Strict separation at resonance]\label{app-sr:gap}
Suppose $\alpha(1-\alpha)=rV_q=\sigma^2>0$.
Fix $0<\theta<1/2$, write $\kappa=\kappa_\theta$, and put
\begin{equation}\label{app-sr:g}
 M=\Phi(1+\kappa/2)-\Phi(1-\kappa/2),\qquad
 g=\frac{\kappa M}{2(2+\kappa)}>0.
\end{equation}
Then, for every fixed real $d$,
\begin{equation}\label{app-sr:error-gap}
 \liminf_n e_{\SR}(n,m_n(d))
 \ge g-\frac{|d|}{\sigma\sqrt{2\pi}}.
\end{equation}
In particular, for $\epsilon=g/4$,
\begin{align}\label{app-sr:rate-gap}
 \limsup_n\frac{m_{\SR}^*(n,\epsilon)-rn}{\sqrt n}
 &\le-\frac{g\sigma\sqrt{2\pi}}{2C_q}<0,\nonumber\\
 \lim_n\frac{m_{\NS}^*(n,\epsilon)-rn}{\sqrt n}
 &=\frac{2\sigma}{C_q}
 \Phi^{-1}\!\left(\frac{1+\epsilon}{2}\right)>0.
\end{align}
\end{proposition}
\begin{proof}
For the visible-bit count $K$ and reconstruction $Z$, set
\[
 J_q(S,Z)=mC_q-\ell_q\{d_H(S,Z)-mq\},\qquad
 X_n=\frac{K-n\alpha}{\sqrt n},\quad
 Y_n=\frac{J_q(S,Z)-n\alpha}{\sqrt n}.
\]
Then $X_n\Rightarrow N(0,\sigma^2)$, whereas the ideal target score
converges to $\Gamma_d=N(d,\sigma^2)$. For each fixed output $z$,
$\E_S2^{J_q(S,z)}=1$. Given a mask with $K=k$ and the seed, a
deterministic decoder has at most $2^k$ reconstruction words.
A union bound and Markov's inequality therefore give
\[
 \Pr\{J_q(S,Z)>k+u\mid\text{mask,seed}\}\le2^{-u}.
\]
Allowing the score to take values on the extended real line, pass to
a subsequential joint limit $(X,Y)$. The preceding bound implies
$X\sim N(0,\sigma^2)$ and $Y\le X$ almost surely. The limit may
have mass at $Y=-\infty$. Data processing and lower semicontinuity
of total variation give
\begin{equation}\label{app-sr:limit-error}
 e:=\TV(\mathcal L(Y),\Gamma_d)
 \le\liminf\{\text{protocol errors on the subsequence}\}.
\end{equation}

We next pass from the finite broadcast bound to a constraint on this
joint limit. For counts $k_i=n\alpha+x_i\sqrt n$ and score thresholds
$y_i$, choose distortion bounds corresponding to
$J_q\ge n\alpha+y_i\sqrt n$. For each mask, use the original
protocol's decoder, ignoring the other mask. Each marginal failure
probability then depends only on its own count, as required for the
averaging below. Taylor expansion of the left side of
\eqref{app-sr:finite}, divided by $\sqrt n$, gives
\[
 A_\theta=a_\theta y_1+(1-a_\theta)y_2
              -f_\theta x_1-(1-f_\theta)x_2.
\]
The terms involving $d$ cancel. The right side, divided by $\sqrt n$,
converges to $\eta(x_2-x_1)/(1-\eta)$. Consequently the sum of the
two marginal failure probabilities is at least
\begin{equation}\label{app-sr:pair}
 \frac{[A_\theta]_+}{x_2-x_1+[A_\theta]_+}.
\end{equation}
We use this bound after averaging over two separated count intervals.
No pointwise convergence of conditional laws is needed. If simultaneous
success has probability zero at a count pair, the sum of marginal
failure probabilities is at least one, and the bound is immediate.
For all other count pairs, apply
\eqref{app-sr:finite} to each finite count pair and invert its
increasing right side in $\eta$. The joint failure probability is
at most the sum of the two marginal failure probabilities. The inverse
converges to \eqref{app-sr:pair} uniformly on compact separated
intervals, since the Taylor estimates are uniform there and the
limiting inverse stays below one. Now average the inequality over
the product of the two normalized count distributions. The CLT and
Portmanteau theorem pass this averaged bound to the joint limit:
write it as an upper bound on the sum of success probabilities,
apply Portmanteau to the open events $Y>y_i$, and take complements.
The count-interval boundaries have zero Gaussian probability.
The argument permits the encoder to depend arbitrarily on blocklength.

Now take
\[
 y_1=-\sigma-\kappa\sigma/2,\qquad
 y_2=\sigma-\kappa\sigma/2,\qquad
 I_i=[y_i,y_i+\kappa\sigma].
\]
Both count intervals have probability $M$ in the limit. If
$x_i=y_i+\delta_i\in I_i$, then
\begin{align*}
 A_\theta-\frac\kappa2(x_2-x_1)
 =\kappa\sigma
 -(a_\theta+\kappa/2)\delta_1
 -(1-a_\theta-\kappa/2)\delta_2\ge0.
\end{align*}
The two coefficients are positive and sum to one. The averaged
version of \eqref{app-sr:pair} therefore gives
\[
 \Pr\{Y\le y_1\mid X\in I_1\}
 +\Pr\{Y\le y_2\mid X\in I_2\}
 \ge\frac\kappa{2+\kappa}.
\]
Since $Y\le X$, we have
\[
 \Pr\{Y\le y<X\}=F_Y(y)-F_X(y)
 \le e+\sup_z|F_{\Gamma_d}(z)-F_{\Gamma_0}(z)|
 \le e+\frac{|d|}{\sigma\sqrt{2\pi}}.
\]
Each conditional failure probability on the preceding two intervals
is at most this quantity divided by $M$.
This proves \eqref{app-sr:error-gap}, using
\eqref{app-sr:limit-error} on a subsequence approaching the liminf.
At $d=-g\sigma\sqrt{2\pi}/2$, the error lower bound is
$g/2>\epsilon$. Target marginalization implies the SR limsup bound
in \eqref{app-sr:rate-gap}. The NS value follows from resonance in
Theorem~\ref{thm:main}.
\end{proof}

For a concrete binary pair, take $q=1/100$,
$\alpha=1-V_q/C_q$ and $\theta=1/10$. Then
$\alpha(1-\alpha)=rV_q$ and
\[
 T=\mathrm{BEC}(1-\alpha),\quad
 W=\mathrm{BSC}(0.01),\quad
 1-\alpha\simeq0.473,\quad r\simeq0.573,
\]
with $g\simeq1.857\times10^{-3}$. At
$\epsilon=g/4\simeq4.643\times10^{-4}$, the SR coefficient is at most
$-1.264\times10^{-3}$, whereas the NS coefficient is approximately
$6.321\times10^{-4}$.
The constants have not been optimized. The result gives a strict
second-order separation without assuming that an SR coefficient
exists or determining its optimal value. The unrestricted EA
coefficient for this pair remains open. The SR converse uses the
requirement of one encoder for different noise realizations, as in
broadcast converses for source--channel coding
\cite{kochman2020distortion}.

\section{Exact conversion of additive channels}
\label{app:additive}
For additive channels, the binary-experiment description is exact
at every blocklength and requires no auxiliary communication.
Let $G,H$ be finite abelian groups, let $P,R$ be noise distributions
on these groups, and define
$T_P(y\mid x)=P(y-x)$ and $W_R(z\mid s)=R(z-s)$.
Then
\begin{equation}
 \begin{split}
 e_{\NS}(n,m;T_P,W_R)
 &=\min_{M:\,MU_{G^n}=U_{H^m}}
                         \TV(MP^{\otimes n},R^{\otimes m})\\
 &=\sup_{0\le y\le1}
 \{\alpha_y(R^{\otimes m},U_{H^m})
                   -\alpha_y(P^{\otimes n},U_{G^n})\}.
 \end{split}
 \label{eq:additive}
\end{equation}
Here the minimum is over stochastic matrices, and $U$ denotes the
uniform law on the indicated group. This formula applies to every
noise distribution, including those with zero dispersion or zero
capacity. Positive capacities are needed only when we express the
asymptotic rate as a capacity ratio.

To prove the first equality, start with a stochastic map $M$ that
preserves the uniform reference and set
\[
 N(x,z\mid s,y)=|G|^{-n}M_{z-s,y-x},
 \qquad x,y\in G^n,\quad s,z\in H^m.
\]
Stochasticity gives the sender marginal $|G|^{-n}$, while uniform
preservation gives the receiver marginal $|H|^{-m}$. Thus $N$ is
non-signalling. Contracting it with the resource channel gives an
additive channel with noise $MP^{\otimes n}$. Every row therefore
has the same total variation error.

Conversely, average any NS channel independently over simultaneous
translations of $(x,y)$ and $(s,z)$. The channel laws are invariant
under these translations, and convexity of total variation ensures
that the error cannot increase. The averaged channel depends only
on $y-x$ and $z-s$, so the NS conditions force both marginals to
be uniform.
Setting $M_{b,a}=|G|^nN(x,z\mid s,y)$ for $a=y-x,b=z-s$
recovers a stochastic map that preserves the uniform reference.
This proves the first equality in \eqref{eq:additive}. The
testing-curve identity \eqref{eq:testing-gap} proves the second.

The uniform input achieves capacity, and
\[
 C_{T_P}=\log|G|-H(P),\qquad
 V_{T_P}=\Var_{A\sim P}\log(|G|P(A)).
\]
Every capacity-achieving input induces a uniform output. The
information density in each row consequently has the distribution
$\log(|G|P(A))$ for $A\sim P$. Hence the minimum and maximum
dispersions agree, even when the capacity-achieving input is not
unique. For BSCs, \eqref{eq:additive} compares binomial testing
curves, which can be evaluated exactly by randomizing at a threshold.
When the capacities and dispersions are positive,
Theorem~\ref{thm:main} gives their Gaussian limit.

\subsection{Zero-dispersion additive channels}
The exact characterization \eqref{eq:additive} does not require
positive dispersion. An additive channel has $V_T=0$ if and
only if its noise is uniform on its support $S\subseteq G$.
In this case $C_T=\log(|G|/|S|)$ and
\begin{equation}
 \alpha_y(P^{\otimes n},U_{G^n})
       =\min\{1,2^{nC_T}y\},\qquad 0\le y\le1.
 \label{eq:zero-additive-testing}
\end{equation}
To see this, zero variance means that $\log(|G|P(a))$ is constant
on the support of $P$. For the product distribution, the likelihood-ratio is therefore $2^{nC_T}$ on $S^n$ and zero elsewhere. This
gives \eqref{eq:zero-additive-testing}.
Such channels need not be noiseless: uniform noise on $\{0,1\}$
in $\mathbb Z_3$ gives capacity $\log(3/2)>0$ and zero dispersion.

If both additive channels have zero dispersion and positive capacities,
substitution in \eqref{eq:additive} gives the exact formulas
\begin{equation}
 \begin{aligned}
 e_{\NS}(n,m;T,W)&=[1-2^{nC_T-mC_W}]_+,\\
 m^*_{\NS}(n,\epsilon;T,W)
   &=\left\lfloor\frac{nC_T-\log(1-\epsilon)}{C_W}\right\rfloor.
 \end{aligned}
 \label{eq:zero-additive-exact}
\end{equation}
The second-order coefficient is therefore zero.

If one channel is additive with zero dispersion and the other has
positive dispersion, we can reduce the problem to coding or
simulation through a noiseless channel. Write $K$ for the
zero-dispersion channel. Comparing
\eqref{eq:zero-additive-testing} with the testing curve of a noiseless
alphabet gives the exact NS conversions
\[
 \id_{\lceil2^{kC_K}\rceil}\ \longrightarrow\ K^{\otimes k}
 \ \longrightarrow\ \id_{\lfloor2^{kC_K}\rfloor}.
\]
The logarithms of these two alphabet sizes differ from $kC_K$ by
at most $O(1)$. The usual coding and simulation expansions
\cite{polyanskiy2010channel,matthews2012linear,cao2024channel}
then yield, with $r=C_T/C_W$ and $C_T,C_W>0$,
\begin{equation}
 m^*_{\NS}(n,\epsilon;T,W)=rn+
 \begin{cases}
  \dfrac{\sqrt{V_T}}{C_W}\Phi^{-1}(\epsilon)\sqrt n+O(\log n),
       &V_W=0<V_T,\\[2mm]
  \dfrac{\sqrt{rV_W}}{C_W}\Phi^{-1}(\epsilon)\sqrt n+O(\log n),
       &V_T=0<V_W.
 \end{cases}
 \label{eq:zero-additive-mixed}
\end{equation}
The positive-dispersion channel need not be additive, provided its
capacity-achieving dispersion is unique. The zero-dispersion channel
must be additive. Composing coding or simulation protocols with
the displayed exact conversions gives both the upper and lower
bounds. We therefore do not need to substitute zero variance into
a Gaussian estimate.

When the zero-dispersion channel is $\id_2$, the coding and simulation
theorems give the same second-order expansion for all three
assistance classes
\cite{polyanskiy2010channel,matthews2012linear,cao2024channel}
under the positive unique-dispersion assumption on the noisy channel:
for $A\in\{\SR,\EA,\NS\}$,
\begin{align*}
 m_A^*(n,\epsilon;T,\id_2)
 &=nC_T+\sqrt{nV_T}\Phi^{-1}(\epsilon)+O(\log n),\\
 m_A^*(n,\epsilon;\id_2,W)
 &=\frac n{C_W}
 +\frac{\sqrt{V_W}}{C_W^{3/2}}\Phi^{-1}(\epsilon)\sqrt n+O(\log n).
\end{align*}
Together, these statements settle the zero-dispersion cases for
additive channels. They do not settle the general nonadditive
boundary $V_{\min}=0$.

\subsection{Finite-blocklength certificates and explicit remainder bounds}
\label{app:finite-resonance}
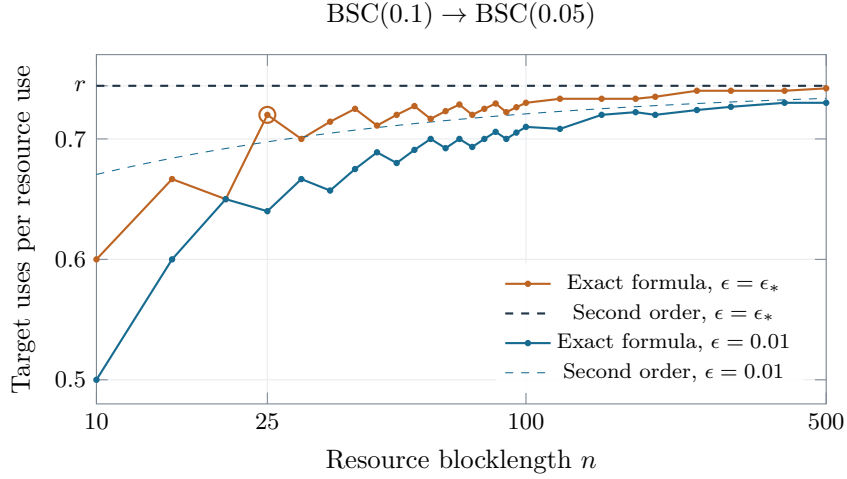
\begin{figure}[t]
\centering
\begin{tikzpicture}
\begin{axis}[width=.72\linewidth,height=6.2cm,
 title={ $\mathrm{BSC}(0.1)\to\mathrm{BSC}(0.05)$},
 title style={font=\small},
 xmode=log,xmin=10,xmax=500,ymin=.48,ymax=.77,
 xlabel={Resource blocklength $n$},ylabel={Target uses per resource use},
 xtick={10,25,100,500},xticklabels={10,25,100,500},
 ytick={.5,.6,.7,.744117352},
 yticklabels={$0.5$,$0.6$,$0.7$,$r$},
 tick label style={font=\footnotesize},label style={font=\small},
 axis line style={ink!60},tick style={ink!60},
 grid=major,grid style={ink!10},
 legend style={at={(.97,.03)},anchor=south east,font=\scriptsize,
 draw=none,fill=white,fill opacity=.9,text opacity=1},
 clip=false]
 \addplot[plotorange,thick,mark=*,mark size=.8pt] coordinates {
 (10,.6) (15,.666666667) (20,.65) (25,.72) (30,.7)
 (35,.714285714) (40,.725) (45,.711111111) (50,.72)
 (55,.727272727) (60,.716666667) (65,.723076923)
 (70,.728571429) (75,.72) (80,.725) (85,.729411765)
 (90,.722222222) (95,.726315789) (100,.73) (120,.733333333)
 (150,.733333333) (180,.733333333) (200,.735) (250,.74)
 (300,.74) (400,.74) (500,.742)
 };
 \addlegendentry{Exact formula, $\epsilon=\epsilon_*$}
 \addplot[ink,dashed,thick] coordinates {(10,.744117352) (500,.744117352)};
 \addlegendentry{Second order, $\epsilon=\epsilon_*$}
 \addplot[plotblue,thick,mark=*,mark size=.8pt] coordinates {
 (10,.5) (15,.6) (20,.65) (25,.64) (30,.666666667)
 (35,.657142857) (40,.675) (45,.688888889) (50,.68)
 (55,.690909091) (60,.7) (65,.692307692) (70,.7)
 (75,.693333333) (80,.7) (85,.705882353) (90,.7)
 (95,.705263158) (100,.71) (120,.708333333)
 (150,.72) (180,.722222222) (200,.72) (250,.724)
 (300,.726666667) (400,.73) (500,.73)
 };
 \addlegendentry{Exact formula, $\epsilon=0.01$}
 \addplot[plotblue,dashed,domain=10:500,samples=60]
 {.744117351665-.232761412567/sqrt(x)};
 \addlegendentry{Second order, $\epsilon=0.01$}
 \addplot[only marks,mark=o,mark size=3pt,plotorange,thick]
 coordinates {(25,.72)};
\end{axis}
\end{tikzpicture}
\caption{NS finite-blocklength rates evaluated numerically from the exact formula for
$\mathrm{BSC}(0.1)\to\mathrm{BSC}(0.05)$, joined as guides to the eye,
and second-order rate approximations. At
$\epsilon_*\simeq0.042$, the dashed prediction is the capacity
ratio $r$ itself. The circled point is the exact $n=25,m=18$
certificate. Integer and higher-order effects remain visible.
The Gaussian approximation is not a bound.}
\label{fig:finite-bsc}
\end{figure}
For $0<p<1/2$, let $A_{n,p}(u)$ be the piecewise-linear curve joining
$(0,0)$ and
\[
 u_{n,j}=2^{-n}\sum_{i=0}^j\binom ni,\qquad
 A_{n,p}(u_{n,j})=F_{n,p}(j):=
 \sum_{i=0}^j\binom ni p^i(1-p)^{n-i},\quad 0\le j\le n.
\]
This is the testing curve of
$(\operatorname{Ber}(p)^{\otimes n},U_{\{0,1\}^n})$.
Its linear segments correspond to randomization within a Hamming
shell.
Consequently, \eqref{eq:additive} gives
\begin{equation}
 e_{\NS}(n,m;\mathrm{BSC}(p),\mathrm{BSC}(q))
 =\max_{0\le j\le m}
 \{F_{m,q}(j)-A_{n,p}(u_{m,j})\}.
 \label{eq:bsc-finite-testing}
\end{equation}
It suffices to check the target knots. Between two consecutive knots,
the target curve is affine and the source curve is concave, so their
difference is convex and attains its maximum at an endpoint.
For rational $p,q$, every quantity in \eqref{eq:bsc-finite-testing}
is rational. The resulting error bounds are therefore exact
finite-blocklength certificates, rather than asymptotic approximations.

For $p=1/10$, $q=1/20$, $n=25$ and $m=18$, the maximum occurs at
$u=2^{-18}$. The target test accepts only the zero-error word. The
source test accepts shells $0,1$ and $102/300$ of shell $2$. Hence
\begin{equation}
 \begin{split}
 e_{\NS}(25,18)
 &=\left(\frac{19}{20}\right)^{18}
   -\frac{102}{25}\left(\frac9{10}\right)^{23}
 \in(0.035,0.036).
 \end{split}
 \label{eq:short-block-certificate}
\end{equation}
Rational bounds on logarithms give $18<25C_p/C_q<19$ and
$\sqrt{V_p/(rV_q)}>1.19$. The latter inequality also shows that
$\epsilon_*>0.04$: evaluating the Gaussian deficit at
$\xi=-\sqrt{V_p}$ and using convexity of the normal density $\phi$
on $[1,1.19]$ gives
\[
 \epsilon_*\ge\int_1^{1.19}\phi(t)\,dt
 \ge0.19\phi(1.095)>0.04.
\]
Thus the target length $m=18=\lfloor25r\rfloor$ is achievable at
the exact cancellation error $\epsilon_*$, without relying on
rounded Gaussian probabilities. Integer and higher-order corrections
need not have a favourable sign at every blocklength: for the same pair,
$e_{\NS}(50,\lfloor50r\rfloor)=e_{\NS}(50,37)>0.069$.

For $p=2/5$, $q=9/20$, $n=20$ and $m=81$, the maximum in
\eqref{eq:bsc-finite-testing} occurs at $u=u_{81,34}$. Put
\[
 \theta=\frac{u_{81,34}-u_{20,6}}{u_{20,7}-u_{20,6}}.
\]
Then
\[
 \begin{split}
 e_{\NS}(20,81)
 &=F_{81,9/20}(34)-F_{20,2/5}(6)
   -\theta\binom{20}{7}(2/5)^7(3/5)^{13}\\
 &\in(0.007,0.008).
 \end{split}
\]
Since $80<20C_p/C_q<81$, this example achieves the capacity-ratio
length rounded \emph{up}.
The accompanying script \texttt{verify\_resonance.py} checks both
maxima and all displayed rational brackets by exact arithmetic.

For $0<p,q<1/2$, we can also obtain a Gaussian upper bound on the
finite-blocklength error with explicit constants. Define
\[
 \Delta=mC_q-nC_p,\qquad
 B_p=0.4748\,\frac{p^2+(1-p)^2}{\sqrt{p(1-p)}},
\]
\[
 G_{n,m}(a)=\sup_{t\in\mathbb R}
 \left\{\Phi\!\left(\frac{t}{\sqrt{nV_p}}\right)
 -\Phi\!\left(\frac{t-a}{\sqrt{mV_q}}\right)\right\}.
\]
For every $n,m\ge1$ and $\gamma>0$,
\begin{equation}
 e_{\NS}(n,m)
 \le\min\left\{1,\,
 G_{n,m}(\Delta+\gamma)
 +\frac{B_p}{\sqrt n}+\frac{B_q}{\sqrt m}+2^{-\gamma}\right\}.
 \label{eq:explicit-gaussian-upper}
\end{equation}
Indeed, at reference budget $u=2^{-h}$, the spectrum bounds
\eqref{eq:spectrum} give
\[
 A_{m,q}(u)-A_{n,p}(u)
 \le\Pr[J_q^{(m)}\ge h-\gamma]
       -\Pr[J_p^{(n)}\ge h]+2^{-\gamma},
\]
where $J_p^{(n)}$ is the sum of $n$ independent information densities
$\log(2\operatorname{Ber}(p)(X))$. Its mean is $nC_p$ and its
variance is $nV_p$. The standardized third absolute moment of one
summand is $(p^2+(1-p)^2)/\sqrt{p(1-p)}$. Applying the iid
Berry--Esseen bound \cite{shevtsova2011constants} to both tails and
setting $t=h-nC_p$ proves \eqref{eq:explicit-gaussian-upper}.
Using left limits of the distribution functions accounts for the
atoms of these lattice laws.

In particular, at $m=\lfloor rn\rfloor\ge1$, the bound implies
\begin{equation}
 \begin{split}
 e_{\NS}(n,m)\le\epsilon_*
 &+\frac{|\ln(m/(rn))|}{2\sqrt{2\pi\mathrm e}}
  +\frac{\gamma}{\sqrt{2\pi mV_q}}\\
 &+\frac{B_p}{\sqrt n}+\frac{B_q}{\sqrt m}+2^{-\gamma}.
 \end{split}
 \label{eq:capacity-error-remainder}
\end{equation}
To obtain this bound, use $\Delta\le0$, the maximum normal density, and
$\sup_x|x\phi(x)|=1/\sqrt{2\pi\mathrm e}$ to control the change
from $mV_q$ to $rnV_q$. For $n\ge2$, taking
$\gamma=\tfrac12\log n$ gives
an explicit $O(\log n/\sqrt n)$ error remainder. This gives a
rigorous asymptotic error bound, but the constants are coarse for
the short blocks considered above. For those blocks, the exact
testing curves give sharper certificates.

\section{Exact finite-blocklength BEC-to-BSC conversion}
\label{app:bec-bsc-exact}
We now derive the exact finite-blocklength formula used for the NS
curves in Figure~\ref{fig:resonance-rates}(a). The proof characterizes
the feasible NS channels themselves, rather than only comparing joint
distributions for a prescribed input law.
Let $K\sim\operatorname{Bin}(n,1-\delta)$ be the number of unerased
resource symbols, and define
\[
 v_j=\sum_{i=0}^j\binom mi,\qquad
 F_{m,q}(j)=\sum_{i=0}^j\binom mi q^i(1-q)^{m-i}
 \quad(0\le j\le m).
\]

\begin{proposition}[Exact BEC-to-BSC conversion]
\label{prop:bec-bsc-exact}
For $n,m\ge1$, $0\le\delta\le1$, and $0\le q\le1/2$,
\begin{equation}
 e_{\NS}(n,m;\mathrm{BEC}(\delta),\mathrm{BSC}(q))
 =\max_{0\le j\le m}
 \left\{F_{m,q}(j)-\E\min\{1,v_j2^{K-m}\}\right\}.
 \label{eq:bec-bsc-finite-testing}
\end{equation}
The $j=m$ term is zero, so the right side is nonnegative.
The proof constructs an explicit NS channel for every feasible noise
law, without auxiliary communication.
\end{proposition}

\begin{proof}
Write $x\in\{0,1\}^n$ for the resource input, $s,z\in\{0,1\}^m$
for the target input and output, and $y\in\{0,1,\mathrm e\}^n$
for the resource output. Average an NS channel independently over
source bit flips $(x,y)\mapsto(x\oplus a,a\cdot y)$, with erasures
unchanged, and target translations $(s,z)\mapsto(s\oplus b,z\oplus b)$.
The channel laws are covariant under these operations. Averaging
therefore preserves the NS constraints and, by convexity, cannot
increase the worst-row error. If $E$ is the erasure set of $y$ and
$k=n-|E|$, an averaged channel has the form
\[
 N(x,z\mid s,y)=2^{-n}Q_E(z\oplus s\mid u),
 \qquad u=(y_i\oplus x_i)_{i\notin E}\in\{0,1\}^k.
\]
The two NS constraints are equivalent to
\begin{equation}
 \sum_wQ_E(w\mid u)=1,\qquad
 2^{-k}\sum_uQ_E(w\mid u)=2^{-m}.
 \label{eq:bec-pattern-ns}
\end{equation}
The first identity gives the sender marginal $2^{-n}$. To check
the second, each visible word $u$ occurs for exactly $2^{n-k}$
choices of $x$, so the receiver marginal is $2^{-m}$. Conversely,
the symmetries and NS constraints force both marginals to be uniform,
which gives the two identities.

On an actual BEC output, every unerased bit agrees with the input.
The induced channel therefore uses only the row
$R_E(w)=Q_E(w\mid0^k)$. Equation~\eqref{eq:bec-pattern-ns} requires
$R_E(w)\le2^{k-m}$ for every $w$. This pointwise bound is also
sufficient: for $k\ge1$, set
\begin{equation}
 Q_E(w\mid0^k)=R_E(w),\qquad
 Q_E(w\mid u)=\frac{2^{k-m}-R_E(w)}{2^k-1}\quad(u\ne0^k).
 \label{eq:bec-row-completion}
\end{equation}
Each row is a probability distribution, and their uniform average
is uniform on the target alphabet. For $k=0$, the bound forces
$R_E=U_m$, so no further rows are needed. Thus every family satisfying
the pointwise bounds extends to an NS box on all its inputs.
The simulated channel is additive, with noise law
\[
 R=\sum_{E\subseteq[n]}\delta^{|E|}(1-\delta)^{n-|E|}R_E.
\]
Patterns with the same number $k$ of visible bits have the same
pointwise bound. We may average their rows without changing the
noise law. Conversely, we may assign the same row to all such
patterns. The feasible noise laws are therefore exactly
\begin{equation}
 R=\sum_{k=0}^n\lambda_kR_k,\qquad
 \lambda_k=\Pr\{K=k\},\qquad
 R_k\in\mathcal P(\{0,1\}^m),\quad R_k(w)\le c_k:=2^{k-m}.
 \label{eq:bec-capped-mixture}
\end{equation}
Every row of the simulated channel has error $\TV(R,P)$, where
$P(w)=q^{|w|}(1-q)^{m-|w|}$.

It remains to minimize the total variation distance over these
mixtures. For this, use the max-flow/min-cut theorem of
Ford--Fulkerson \cite{ford1956maximal}. Construct a network with
a source, one vertex for each $k$, one vertex for each word $w$,
and a sink. The edges from the source to $k$, from $k$ to $w$,
and from $w$ to the sink have capacities $\lambda_k$,
$\lambda_kc_k$, and $P(w)$, respectively.
The maximum flow equals the largest attainable overlap
$\sum_w\min(P(w),R(w))$ in \eqref{eq:bec-capped-mixture}.
Indeed, any feasible mixture gives a flow by retaining only its
overlap with $P$. Conversely, extend the mass leaving each vertex
$k$ to a total of $\lambda_k$, without exceeding the capacities
of its edges to the word vertices. This is possible because
$2^mc_k=2^k\ge1$. The completed rows give a feasible mixture whose
overlap with $P$ is at least the original flow.

Fix the set $A$ of word vertices on the source side of a cut.
Minimizing over the side chosen for each $k$ vertex gives the cut
capacity
\[
 P(A)+\sum_k\lambda_k\min\{1,c_k|A^c|\}.
\]
The max-flow/min-cut theorem therefore proves
\begin{equation}
 \min_R\TV(P,R)
 =\max_{B\subseteq\{0,1\}^m}
   \left\{P(B)-G(|B|)\right\},\qquad
 G(l)=\sum_k\lambda_k\min\{1,2^{k-m}l\}.
 \label{eq:bec-lorenz-flow}
\end{equation}
For a fixed cardinality $l$, the largest value of $P(B)$ is the sum
of the $l$ largest noise probabilities. When $q\le1/2$, these
probabilities are ordered by Hamming weight. Between consecutive
shell endpoints $v_j$, the sum is affine in the integer $l$,
whereas $G(l)$ is concave. Their difference therefore attains its
maximum at a shell endpoint, or at $l=0$, where it is zero.
At $l=v_j$, the two terms are $F_{m,q}(j)$ and
$\E\min\{1,v_j2^{K-m}\}$. This proves
\eqref{eq:bec-bsc-finite-testing}.
\end{proof}

For numerical evaluation, each expectation in
\eqref{eq:bec-bsc-finite-testing} needs only a binomial tail and a
truncated exponential moment. With
$a_j=m-\lfloor\log_2v_j\rfloor$,
\[
 \E\min\{1,v_j2^{K-m}\}
 =\Pr\{K\ge a_j\}
   +v_j2^{-m}\sum_{k<a_j}2^k\Pr\{K=k\}.
\]
Thus each error evaluation requires only $m+1$ testing budgets.
Since the error is nondecreasing in $m$, the largest admissible
target length can be found by binary search. The script
\texttt{plot\_resonances.py} reproduces
Figure~\ref{fig:resonance-rates} and checks selected NS points by exact
rational arithmetic at $\delta=435775832558017/10^{15}$.
The exact certificates refer to this rounded rational parameter.
The threshold root used for the plot is determined numerically.

\end{document}